%% file: main.tex
\documentclass[manuscript,screen]{acmart}

\usepackage{amsthm}
\usepackage{bm}
\usepackage{multirow}

\usepackage{tabularx}
\usepackage{subcaption}

\newcommand{\vtr}[1]{\mathbf{#1}}
\newcommand{\norm}[2][2]{\lVert {#2}\rVert_{#1}}
\newcommand{\dprod}[2]{\langle  {#1},  {#2} \rangle}
\newcommand{\eprod}[2]{  {#1} \odot  {#2}}
\newcommand{\osqr}[1]{{#1^{\odot 2}}}
\newcommand{\float}{\mathsf{fl}}
\newcommand{\tconfig}{T3}
\newcommand{\epsone}{EP1}
\newcommand{\epssharp}{EP\#}
\newcommand{\epsroot}{EPN}

\newtheorem{proposition}{Proposition}
\newtheorem{theorem}{Theorem}
\theoremstyle{definition}
\newtheorem{definition}{Definition}

\usepackage{cleveref}
\crefformat{section}{#2Sec.~#1#3}
\crefformat{figure}{#2Fig.~#1#3}
\crefformat{algorithm}{#2Alg.~#1#3}
\crefformat{table}{#2Tab.~#1#3}
\crefformat{equation}{#2Eq.~#1#3}
\crefformat{definition}{#2Def.~#1#3}
\crefformat{appendix}{#2App.~#1#3}
\crefformat{proposition}{#2Prop.~#1#3}
\crefformat{theorem}{#2Thm.~#1#3}
\crefmultiformat{section}
  {#2Secs.~#1#3}
  { and~#2#1#3}
  {, #2#1#3}
  {, and~#2#1#3}
\crefmultiformat{proposition}
  {#2Props.~#1#3}
  { and~#2#1#3}
  {, #2#1#3}
  {, and~#2#1#3}

\definecolor{teebg}{RGB}{237,245,253}
\definecolor{teeborder}{RGB}{100,135,170}

\definecolor{gpubg}{RGB}{253,222,220}
\definecolor{gpuborder}{RGB}{225,185,180}

\definecolor{mygreen}{RGB}{34,139,34}
\definecolor{myblue}{RGB}{30,100,200}

\newcommand{\qwen}{\texttt{Qwen-4B}}
\newcommand{\llama}{\texttt{Llama-3B}}
\newcommand{\suppmat}{our Supplementary Materials\textsuperscript{\ref{footnoteCodeRepo}}}

\usepackage{tikz}
\usetikzlibrary{positioning,fit}
\usetikzlibrary{arrows.meta}
\usetikzlibrary{shapes.geometric}

 \renewcommand\footnotetextcopyrightpermission[1]{}

\begin{document}

\title[SpliTEE]{SpliTEE: Fast and Private LLM Inference by Coupling GPU-Assisted Trusted Execution Environments with Differential Privacy}

\author{Shashie Dilhara Batan Arachchige}
\email{shashiedilhara.batanarachchige@hdr.mq.edu.au}
\orcid{0009-0008-3660-550X}
\author{Robin Carpentier}
\authornote{Both authors contributed equally to this research.}
\email{robin.carpentier@mq.edu.au}
\orcid{0000-0003-1369-2248}
\author{Hassan Jameel Asghar}
\authornotemark[1]
\email{hassan.asghar@mq.edu.au}
\orcid{0000-0001-6168-6497}
\author{Dali Kaafar}
\email{dali.kaafar@mq.edu.au}
\orcid{0000-0003-2714-0276}
\affiliation{%
  \department{School of Computing}
  \institution{Macquarie University}
  \city{Sydney}
  \state{New South Wales}
  \country{Australia}
}

\renewcommand{\shortauthors}{Batan Arachchige et al.}

\begin{abstract}

User prompts provided to large language models (LLMs) may contain private information. 
One way to protect them is to execute the LLM inside a trusted execution environment (TEE). However, this results in slow inference times as current TEEs are significantly slower than GPUs for LLM inference. To circumvent this, Tram{\`e}r and Boneh (2019) proposed ``Slalom'' which splits neural network inference between a TEE and an untrusted GPU. They encrypt inputs to computations outsourced to the GPU. In this paper, we extend this split-inference architecture to LLM inference and instead protect intermediate inputs using differential privacy (DP). We first demonstrate that \emph{masking} intermediate representations is necessary by showing an 80\% accuracy on a prompt-reconstruction attack from these representations. Our main contribution is a global sensitivity analysis of key functions in LLM inference, which bounds the required scale of DP noise. Unlike encryption, DP avoids quantization, allowing the LLM to remain in the floating-point domain. We also derive an upper bound on the floating-point error from masking and subsequent noise cancellation as a function of the privacy parameter $\epsilon$, keeping the same quality of the LLM response.
We implement our architecture using the Intel TDX TEE and two LLMs: \texttt{Llama-3.2-3B} and \texttt{Qwen3-4B}. 
Our split execution is nearly twice as fast as fully TDX-based inference. Moreover, it is at most 43\% faster than Slalom while achieving higher accuracy. Finally, we demonstrate that prompt reconstruction, even with knowledge of the DP mechanism, cannot recover more information than is contained in an unrelated prompt.
\end{abstract}

\maketitle

\section{Introduction}
The use of artificial intelligence has proliferated in the last couple of years due to the ever-expanding and improving capabilities of large language models (LLMs). Widely used models such as ChatGPT or Claude\footnote{\url{https://chatgpt.com/} and \url{https://claude.ai/}} are publicly accessible via websites or APIs. 
At the same time there are growing privacy concerns due to the remote deployment of these LLMs. Users submit their requests to these LLMs using prompts that may contain sensitive information about individuals and businesses, e.g., sexual orientation, mental health information, code files or clients data~\cite{cyberhaven2023ChatGPT}. LLM \emph{service providers} (e.g., OpenAI) receive and process these prompts.
Naturally, one would not want this information to be misused.
For instance, it is routine practice to leverage user prompts and their responses to train LLMs. As a result, sensitive information contained in these prompts can be memorized by LLMs and subsequently used in answers to other users~\cite{huang2024llm-memorization,carlini2019secret}.
While some providers allow to opt-out of such training, some data or metadata may still be kept~\cite{king2025optout}.

Additionally, collection of vast amounts of personal data by LLM service providers can raise significant concerns, as this data can be sold to advertisers to manipulate what people buy~\cite{zuboff2019age}, shared with third parties to manipulate what people think~\cite{zuboff2019age}, or compromised and leaked to malicious actors~\cite{verizon2026dbir}. Depending on the provider's jurisdiction, data sovereignty concerns may also arise as governments can request access to this data~\cite{cloudact2018}, even when it's stored overseas, and may use it for mass surveillance programs of the population~\cite{poitras2014citizenfour}.

A simple solution for private LLM inference would be to deploy LLMs locally on users' devices. However, this deployment model is not suitable for LLM service providers who want to protect their intellectual property, i.e., the learned weights of the model, which is the case for the current best-performing LLMs~\cite{chiang2024chatbot}. Moreover, hosting such models locally is infeasible due to their substantial computational requirements.
Other solutions for private LLM inference include modifying the prompt to remove sensitive information, or carrying out LLM inference in the encrypted domain using homomorphic encryption or secure multiparty computation. Modifications of the prompt can be performed automatically~\cite{feyisetan2020Privacy,dp-bart}. 
For example, a differentially private solution probabilistically replaces either individual words or sentences in the prompt with semantically similar counterparts (see, for example~\cite{dp-bart}). This is not ideal for utility as it perturbs words in the prompt which may result in lower result quality and, depending on the task, the user may want to keep the text in the prompt unchanged, e.g., specific sales numbers in an Information Extraction task. Cryptographic solutions through homomorphic encryption and secure multiparty computations are computationally or communication-wise expensive, require changes to the LLM to make it more suitable for cryptographic operations, and some part of the LLM computation is assumed to be done publicly, e.g., by the client~\cite{al2026sokfhe, andreoletti2026pp-llm}. 

Another solution, which is also related to the theme of this paper, is to execute the LLM inside a trusted execution environment (TEE). The service provider hosts the LLM within the TEE, with the guarantee that computations performed inside the TEE can neither be observed nor be tampered with by the service provider.
These are standard assumptions for computations performed within a TEE~\cite{costan2016Intel,chengTDXDemystified2024}.
Unfortunately, end-to-end inference within the TEE is slow, as current TEEs are primarily CPU-based which are significantly slower than GPUs for LLM inference. While the first commercial GPU-TEE~\cite{dhanuskodi2023Creating} was released recently, specifications are scarce and the threat model is weaker.%

To solve this issue, researchers have proposed a split-inference architecture, where the expensive operations during LLM inference, namely matrix multiplications, are outsourced to a faster but untrusted GPU on the same machine, and the rest of the operations are carried out inside the TEE~\cite{tramer2018slalom}.\footnote{Although the scheme in~\cite{tramer2018slalom} is built for deep neural networks, it can be adapted for LLM inference, as is done in~\cite{xue2025Securing}.} To prevent information leakage, the inputs to the outsourced components are \emph{masked} by additively blinding them with random nonces. The GPU performs the multiplication as if it received the real input. Due to the linearity of matrix multiplications, the original matrix product can be retrieved from the result returned by the GPU by subtracting the randomized component. Moreover, this can be partly preprocessed, and hence expensive operations are not done at the slower TEE. An issue with the approach used in~\cite{tramer2018slalom} is that the matrix product needs to be \emph{quantized} to work in finite fields, which results in accuracy loss.  

In this paper we look at the LLM split-inference architecture from another angle: instead of encrypting the outsourced inputs via a stream cipher as is done in Slalom, the name of the system proposed in~\cite{tramer2018slalom}, we add differentially private noise calibrated to the sensitivity of the operation. The advantage is that we remain in the original floating-point arithmetic domain in which the LLM is implemented. Since our goal is prompt privacy, differential privacy is a reasonable notion to apply in the split-inference setting. However, as we show, we cannot add unbounded differentially private noise without running into floating-point error accumulation issues. The good news is that we can tune the privacy parameter $\epsilon$ so that error never passes the threshold to impact accuracy. Our method is faster than Slalom~\cite{tramer2018slalom}, and also incurs no accuracy penalty depending on the privacy level $\epsilon$. \cref{fig:architecture}
shows the high-level overview of our scheme.

\begin{figure}
\centering
\resizebox{0.70\columnwidth}{!}{%
\input{Figures/system-architecture}
}
\caption{Overview of the LLM deployment architecture. The TEE is trusted whereas the GPU is honed-but-curious.}
\label{fig:architecture}
\end{figure}

To summarize, we make the following contributions\footnote{Code and Supplementary Materials available at: \url{https://github.com/DPVault/SpliTEE}\label{footnoteCodeRepo}}:
\begin{enumerate}
    \item We present a split-architecture, where LLM inference is split between: i) a trusted but slow TEE, and ii) a honest-but-curious and fast GPU. The computationally heavy operations are outsourced to the GPU as in~\cite{tramer2018slalom}, but we protect their inputs by using differentially private noise such that the service provider cannot infer the user prompt. The GPU being honest-but-curious is a realistic assumption as the service providers are usually concerned about maintaining their reputation.

    \item We identify common components in LLMs that can be algebraically outsourced, using the open-weight model \texttt{Llama-3.2-3B-Instruct} as an illustrative base model. As the inputs to these components are outputs of the preceding components, we derive the global sensitivity bounds of these functions, which measure the required scale of differentially private noise~\cite{dp-book}. Our global sensitivity analysis is a contribution in itself and can be used in any system that outsources part of LLM inference to another entity without disclosing intermediate function inputs and outputs.  

    \item We show that removing differentially private noise from the result computed by the GPU on the outsourced computation results in accumulation of floating-point errors that can negatively impact accuracy of the LLM response. We derive an upper bound on this error for linear operations, and show that it closely matches real execution error. Building on this, we show that the value of the privacy parameter $\epsilon$ can be tweaked so that the error bound remains below a threshold to ensure accuracy is not impacted, while still maintaining solid privacy.
    
    \item We implement our scheme on a machine equipped with Intel Trust Domain Extensions (TDX) and a GPU. Using the LLMs \texttt{Llama-3.2-3B-Instruct} and \texttt{Qwen3-4B-Instruct}, we show that end-to-end CPU-based inference on the protected TDX virtual machine is almost two times slower than our scheme. We further show that our scheme is at most 43\% faster than Slalom~\cite{tramer2018slalom} adapted to LLMs. 

    \item We consider an attacker having access to outsourced inputs and training their own LLM model to reconstruct prompts from said inputs. We show that adding differentially private noise is essential, as the prompt can be reconstructed with almost 80\% accuracy if the inputs are left \emph{non-masked}. We further show that the same attack, even when trained on auxiliary noisy inputs, is unable to infer the prompt significantly better than guessing random tokens expected to be common between unrelated prompts. Our implementation is publicly available.\textsuperscript{\ref{footnoteCodeRepo}}

\end{enumerate}

\section{Related Work}
\label{sec:related-work}

\subsection{Prompt Privacy}
Solutions have been proposed to sanitise LLM prompts before they are sent. These solutions do not require any modification of the LLM pipeline at the service provider's side, only a pre-processing step is added (and potentially post-processing of the answer). For example, \cite{chen2023Hide} proposes to use a Named Entity Recognition model to detect sensitive textual entities and replace them with their category, e.g., \textless DATE\textgreater\ for a date. On the other hand, some solutions replace individual characters in the prompt~\cite{arachchige2026protecting}, while others replace entire words or sentences with semantically similar counterparts~\cite{dp-bart,carpentier2025preempting}. All these solutions necessarily impact utility~\cite{zhang2024Nofreelunch}. In our case, we leverage a TEE on the provider's infrastructure and calibrate DP noise to mask the inputs that are visible to said provider, without degrading utility.

\subsection{TEE-Assisted Private LLM Inference}\label{sec:rel-work-teeLLM}
Slalom~\cite{tramer2018slalom} is one of the earliest approaches to combine a TEE with an untrusted GPU for private and verifiable neural-network inference. Instead of executing the complete model inside the TEE, it outsources expensive linear operations while masking their inputs with a random mask. The corresponding correction is precomputed inside the TEE so that the original linear result can be recovered after GPU execution. Slalom also uses Freivalds' algorithm~\cite{freivalds1977probabilistic} to verify the correctness of the outsourced computation. Our approach %
differs in both the masking mechanism and execution setting. Slalom uses pseudorandom masks over finite-field representations, whereas our approach adds differentially private noise. Slalom was designed for conventional Deep Neural Networks (DNNs) and evaluated using Intel SGX, while our work considers transformer-based LLM inference inside an Intel TDX virtual machine. We further show that our system is better at maintaining the accuracy of LLM responses, as we do not quantize the model which is required in Slalom-based LLM inference. 

SOTER~\cite{shen2022SOTER} also partitions neural-network inference between a TEE and an untrusted accelerator to reduce the cost of executing the complete model inside trusted hardware. Its main objective, however, is to protect the intellectual property of the model when inference is performed on a user device. SOTER identifies \emph{associative} operators and transforms their \emph{parameters} by multiplying them with a random scalar mask before outsourcing the operator to the untrusted GPU; the intended result is then recovered inside the TEE using the corresponding inverse scaling. We could adapt their technique to our use case by 
multiplying intermediate inputs with a vector of random scalars, but this exacerbates the floating-point issues discussed in our paper. 
Furthermore, SOTER considers more general neural-network models instead of Transformer-based LLMs, which has a broader set of operations that must be protected. 

TwinShield~\cite{xue2025Securing} addresses confidential execution of transformer models in an untrusted cloud environment. It aims to protect model parameters, input data, and computation integrity while reducing the amount of computation performed inside the TEE. Unlike earlier approaches that primarily outsource linear operations, TwinShield extends outsourcing to transformer-specific operations, attention matrix multiplication and the non-linear softmax function. It also provides a verification mechanism for the outsourced computation. 
While their mechanism to mask inputs to the linear layers is the same as Slalom i.e., through additive random masks, the non-linear components are protected through a combination of scalar masking and permutation. The former can be proven to be secure assuming the security of the underlying cipher, but the latter is ad hoc and susceptible to algebraic attacks. The authors do not provide any security proof of said combination. Moreover, their scheme suffers from the same accuracy issues as Slalom due to model quantization.

Finally, TEEs have also been employed (e.g., see~\cite{zhang2024No}) to protect model parameters and training data when ran on a user device. This is different from our setting where we want to protect user prompts during inference on the service provider's infrastructure.

\subsection{Embedding Inversion Attacks}\label{sec:rel-embed-inv}
Song and Raghunathan~\cite{song2020information} showed that embeddings can retain enough information to recover words from the original input. Later works extended this to full-sentence reconstruction using optimization-based attacks~\cite{morris2023text} and learning-based approaches that train a generative decoder to recover the input~\cite{li-etal-2023-sentence}. The latter particularly motivates our reconstruction attack. Dong et al.~\cite{dong2025Depth} further study inversion of hidden states exposed during split LLM inference and show that adding Laplace noise directly to these states can reduce reconstruction only at a substantial utility cost. This differs from our setting because their noisy representation is also used for inference. In our design, noise is added only to outsourced inputs, and subsequently removed from the result inside the TEE before inference continues.

\section{Background}
\subsection{Trusted Execution Environments (TEEs)}\label{subsec:tee}

A TEE is a type of hardware enforcing security guarantees during computations. For example, a TEE can provide integrity of the processing code, as well as integrity and confidentiality of the processed data. In some TEEs~\cite{costan2016Intel,chengTDXDemystified2024}, protection against software attacks is achieved by providing a specific execution context, isolated from the rest of the system (including applications, OS and hypervisor). On the other hand, protection against hardware attacks is attained by leveraging encryption. 
TEEs are a powerful tool to perform remote computations privately, where a user submits data (and possibly the processing code) to a remote server for processing. In this case, the user has a proof that the executed code is correct and not tampered with, through \emph{remote attestation}. In a nutshell, a cryptographic hash of the code is signed by the TEE, then received by the user. They can verify the signature and the hash via the corresponding public verification key through the vendor's certificate~\cite{menetrey2022attestation}. 
Intel SGX~\cite{costan2016Intel}, ARM Trustzone~\cite{pintoDemystifyingArmTrustZone2019} and Intel TDX~\cite{chengTDXDemystified2024} are examples of CPU-based TEEs embedded in many commercially available processors. However, our paper deals with LLM inference which is heavily reliant on GPU for speed. Hence, we propose a solution for private LLM inference by coupling a CPU-based TEEs with an untrusted GPU.

While the first commercial GPU-based TEE was recently released, namely NVIDIA Confidential Computing~\cite{dhanuskodi2023Creating}, it is only available on one, very high-end commercial GPU and specifications are scarce, which renders the security analysis of the system difficult~\cite{gu2025NVIDIA}. Some sources mention a weaker threat model, where the VRAM is not encrypted while CPU-based TEEs encrypt RAM to protect against hardware attacks. Nevertheless even GPU-based TEEs are likely to be slower than state-of-the-art GPUs outside the trusted zone.

Note that it may be possible to infer information about data inside TEEs through side channel attacks and meta data exchanged between the TEE and the host machine~\cite{chowdhuryy2024metaleak}. Mitigation of side-channel attacks are orthogonal to our contribution and therefore out of scope.

\subsection{Components of an LLM}
\label{subsec:llm-components}
We describe components common in most LLMs to identify the ones that can be efficiently outsourced. 
We use the open-weight \texttt{Llama-3.2-3B-Instruct} model,\footnote{See \url{https://huggingface.co/meta-llama/Llama-3.2-3B-Instruct}} henceforth only referred to as \llama{}, as the reference model. %
Most LLMs are based on the transformer architecture~\cite{vaswani2017attention}, which divides the computation into the following components:
\begin{enumerate}
    \item An Embedding Model
    \item A Transformer Block (repeated a fixed number of times), itself containing
        \begin{enumerate}
            \item A Multi-Head Self-Attention mechanism
            \item A Feed-Forward Network
        \end{enumerate}
    \item A Language Modeling Head
\end{enumerate}

Computation begins by processing all tokens in the prompt at once, i.e., as one matrix. Once a token is generated after Step 3, it becomes part of the input to generate the next token. After this, Steps 2 and 3 are repeated until the number of desired output tokens have been produced. %
A pictorial representation of these components is shown in \cref{fig:transformer}. In this section, we describe the operators computed by each component, as these details are necessary for the differential privacy analysis presented later.

\begin{figure*}
\centering
\resizebox{\textwidth}{!}{%
\input{Figures/transformer-architecture}
}
\caption{A breakdown of LLM components based on \llama{}. The components offloaded to the GPU are marked with the icon~\raisebox{-0.3ex}{\includegraphics[height=1.2em]{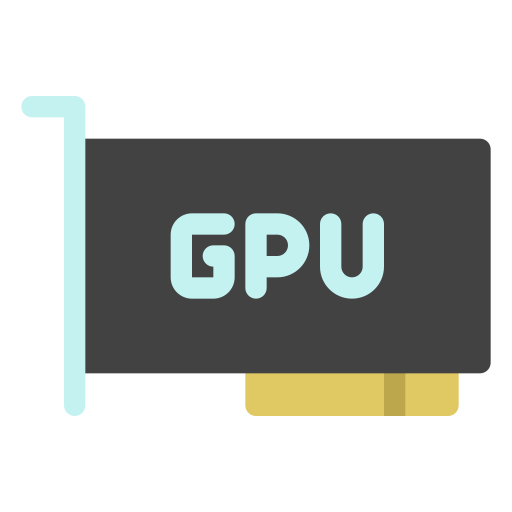}}.}
\label{fig:transformer}
\end{figure*}

\subsubsection{Embedding Model}
The embedding model associates each token of the prompt with an embedding vector of dimension $d$. 
Thus the input prompt can be thought of as the token matrix $X$, which is composed of $n$ tokens $(\vtr{x}_1, \ldots, \vtr{x}_n)$, where each token is a vector in $\mathbb{R}^d$ of embedding dimension $d$ (also called hidden size). In \llama{}, $d = 3072$. 

\subsubsection{Transformer Block -- Self Attention Mechanism}\label{sec:attention}
As illustrated in \cref{fig:transformer}, the first part of the transformer block is the multi-head self-attention mechanism. Roughly, this mechanism calculates how much each token should ``attend'' to every previous token to better understand its context. The self-attention mechanism has the following sub-components.

\paragraph{Input Layer Normalisation.} In \llama{}, the attention mechanism starts with the normalization layer (called, the pre-norm transformer). Each token embedding $\vtr{x}$ is normalized independently using root-mean square normalization (RMSNorm) (see \cref{def:normalization}). This normalization is done at the start of each layer. Each layer has its own learnable parameters $\bm{\alpha}, \bm{\beta} \in \mathbb{R}^d$ (see Eq.~\eqref{eq:ln}). 
After this step, the initial token matrix $X$ is transformed to the normalized token matrix $\widehat{X}$.

\paragraph{Linear Projections (\emph{Q}, \emph{K}, \emph{V})}
Three independent linear projections of the input are computed using learned weight matrices $W_\text{qry}$, $W_{\text{key}}$, and $W_\text{val}$, resulting in the quantities $Q =  W_\text{qry} \widehat{X}$, $K =  W_\text{key} \widehat{X}$, and $V =  W_\text{val} \widehat{X}$. The notation $Q$, $K$, and $V$ stand for query, key, and value, respectively. In \llama{}, the matrices $W_\text{qry}$, $W_{\text{key}}$ and $W_\text{val}$ have dimensions $d_\text{qry} \times d$, $d_{\text{key}} \times d$ and $ d_\text{val} \times d$ respectively, where for example $d_\text{qry} = d_\text{key} = d_\text{val} = 500$. We shall assume that they are equal, and represent the variable dimension by $d_\text{key}$. Thus, matrices $Q$, $K$ and $V$ have dimensions $d_\text{key} \times n$ each.%

\paragraph{Rotary Positional Embeddings (RoPE).} The attention mechanism does not know the order of the tokens in the prompt. This positional information is injected through rotary positional embeddings into $Q$ and $K$. With this information, the model can tell whether a token is positioned earlier or later in the sequence.

\paragraph{Scaled Dot-Product Attention.} This component combines the matrices $Q, K$ and $V$ as 
$H = \text{softmax}\left( \frac{Q^T K}{\sqrt{d_\text{key}}}\right) V^T $

\noindent Note that $Q^TK$ is a matrix of dimensions $n \times n$, since $Q^T$ has dimension $n \times d_\text{qry}$ and $K$ has dimension $d_\text{key} \times n$, and we assume $d_\text{qry} = d_\text{key}$. The output of the softmax function is thus also an $n \times n$ matrix. This is then multiplied by the $n \times d_{\text{key}}$ matrix $V^T$, to produce an $n \times d_\text{key}$ matrix $H$.

\paragraph{Linear Projection (\emph{O})} Prior to this component, the outputs from each of the $h$ attention heads are combined together, resulting in the final matrix $H$ being of size $n \times hd_\text{val} = n \times d$, as $h$ is set as $d/d_\text{key}$.\footnote{See \url{https://huggingface.co/docs/transformers/model_doc/llama}.}
The next component is another linear projection using the (learned) weight matrix $W_\text{oh} \in \mathbb{R}^{d \times d}$ as 
$O = W_\text{oh} H^T$ 
which results in a $d \times n$ matrix $O$.

\paragraph{Residual Addition.}  At this point the original (unnormalized) input $X$ is added to $O$ to obtain 
$Y = O + X$
which is again a $d \times n$ matrix, which is the final output of the attention mechanism.

\subsubsection{Transformer Block -- Feed Forward Neural Network (MLP)}\label{sec:ffnl}
This component has the following subcomponents. 

\paragraph{Post Attention Layer Normalization.}
This is the same operation as input layer normalization from the attention mechanism except that the learned parameters are different (see Eq.~\eqref{eq:ln}). The matrix $Y$ is thus transformed into the matrix $\widehat{Y}$ with RMS-normalized entries. 

\paragraph{Linear Projections (\emph{Gate}, \emph{Up}).} Two independent projections of the input are computed, using the (learned) weight matrices $W_\text{gate}, W_\text{up}\in \mathbb{R}^{d_\text{ff} \times d}$, where $d_\text{ff}$ is the model's feed-forward dimension, e.g., 8,192. This gives rise to the $d_\text{ff} \times n$ matrices.
$G = W_\text{gate}\widehat{Y} \text{ and } U = W_\text{up} \widehat{Y}$

\paragraph{Activation Function.}
In \llama{}, the activation function used is the Sigmoid Linear Unit (SiLU).\footnote{See \url{https://docs.pytorch.org/docs/stable/generated/torch.nn.SiLU.html}} We denote this function by $\mathsf{sig}$, which can be applied element-wise on a matrix (see Eq.~\eqref{eq:def:sigmoid} for the definition). After application of the activation function, we obtain the resulting matrix $B$ as
$B = U \odot G \odot \mathsf{sig}(G)$ 
where $\odot$ denotes element-wise matrix product. Note that since each matrix is of dimension $d_\text{ff} \times n$, the matrix $B$ is also of the same dimension.

\paragraph{Linear Projection (\emph{Down}).} The final part of this component is a linear projection using the (learned) weight matrix $W_\text{down}\in \mathbb{R}^{d \times d_\text{ff}}$ giving us
$D = W_\text{down} B$
which is a $d \times n$ matrix.

\paragraph{Residual Addition} This is where the residual output from the self-attention mechanism, i.e., $Y$, is added to $D$, and thus $D$ is updated as $D \leftarrow D + Y$.

\subsubsection{Language Modelling Head}
\label{app:sub:lang-model-head}
The transformer block described in \cref{sec:attention,sec:ffnl} is run $L$ times. Here $L$ represents the number of attention blocks, also called layers, e.g., 16. The output after each layer is called a \emph{hidden state}. The output after the $L$th layer is called the \emph{last hidden state}. This last hidden state is given to the language modelling head to generate a token. The main tasks involved in this component are as follows.

\paragraph{Final Layer Normalization.} Identical to input layer normalization, with different learned parameters. 

\paragraph{Language Modelling.}

For next token prediction, we only consider the last vector of the last hidden state, i.e., corresponding to the last token. This component uses the model's token vocabulary, denoted $\mathsf{vocab}$. The vocabulary's size differs according to the model. In \llama{}, it is $128,256$. %
Each token in the vocabulary has an associated token embedding in $\mathbb{R}^d$. This layer then performs a dot product of the last hidden vector with the token embedding of each token in the vocabulary. This is essentially a multiplication with the matrix of token embeddings of all tokens in the vocabulary, which we call $W_\text{lm}$, whose weights are again obtained through training. We get a vector of values, one for each token in the vocabulary. The model then applies the softmax function to get a vector of probabilities.
Finally, the next token is chosen based on these probability vectors and a configuration, e.g., always choosing the highest probable token is called greedy decoding.

\subsection{Differences With Other LLMs}\label{sec:other_models}

The architecture described above is based primarily on \llama{}, which is the main model used in our evaluation. However, the proposed outsourcing approach is not intended to depend on a single model architecture. To examine how it can be applied more broadly, we consider other widely used Transformer-based model families and identify architectural differences that are relevant to the outsourced computation. 
In particular, we consider Qwen, which is also included in our experimental evaluation, and also discuss Gemma and Mistral. Since our approach outsources the computationally expensive linear projections while keeping the other model operations within the TEE, the main question is whether these models retain similar linear operations and how the surrounding computations differ.

Qwen models follow a similar decoder-only Transformer architecture to Llama, but some design choices differ between model generations. For example, \texttt{Qwen2.5-3B} uses bias terms in the $Q$, $K$, and $V$ projections, whereas \texttt{Qwen3-4B} uses bias-free projections. Here, a bias is an additional learned value that helps the projection adjust its output beyond what is produced by the matrix multiplication alone. For a biased linear layer with output $W\vtr{x}+\vtr{b}$, where $\vtr{b}$ is the bias vector, outsourcing the masked input produces $W(\vtr{x}+\bm{\eta})+\vtr{b}$. Subtracting the masking correction $W\bm{\eta}$ inside the TEE recovers $W\vtr{x}+\vtr{b}$, so no additional correction is required for the bias term. Qwen3 also applies RMSNorm separately to the $Q$ and $K$ representations after their respective projections. This additional normalization remains inside the TEE and does not prevent outsourcing linear projections.

For the Gemma family%
, we consider \texttt{Gemma 3}~\cite{gemma2025}. It retains the same main linear projections in the attention and feed-forward blocks, but introduces several differences in the surrounding operations. Gemma 3 applies additional RMSNorm operations, also normalizes the $Q$ and $K$ representations after projection. It combines local sliding-window attention with global attention~\cite{gemma2025}. Sliding-window attention limits each token to attending within a nearby region, while global attention allows information to be shared across the full sequence. Its feed-forward network also uses a GELU-based gated activation instead of SiLU. These differences affect how the projection outputs are processed, but the main attention and feed-forward projections remain linear.

Finally, we consider Mistral~\cite{jiang2023mistral7b}. For this comparison, we inspect \texttt{Mistral-7B-Instruct-v0.3}. This model uses bias-free linear projections, follows the same general RMSNorm placement as \llama{}, and uses a SiLU-based gated feed-forward block. It also does not apply additional normalization to the query or key representations. Unlike the original Mistral 7B release~\cite{jiang2023mistral7b}, this does not use sliding-window attention.

Overall, our evaluation focuses primarily on \llama{}, with Qwen also included in the experimental evaluation. Gemma and Mistral are discussed here only to examine how the approach extends to other model families. Differences in hidden size and projection dimensions do not affect the basic masking approach, since the same mechanism is applied using the dimensions of each model's linear projections. Model specific architectural differences mainly change the operations surrounding the outsourced projections rather than the masking principle itself. Therefore, the proposed mechanism can be applied in the same way to other models that contain linear projections suitable for outsourcing.

\subsection{Differential Privacy}
Differential privacy is a formal notion of privacy for databases introduced by Dwork et al.~\cite{dwork2006calibrating}. For this paper, we use a variant of differential privacy called $f$-differential privacy~\cite{dong2022gaussian}, as the noise addition mechanism in our scheme, called the Gaussian mechanism, has a simpler characterization under this definition than the usual definition of approximate differential privacy~\cite{dwork2006calibrating}. Our definition is tailored to our use case, which considers privacy with respect to prompts with $n$ tokens. 

Let us begin with some notation. We denote vectors from $\mathbb{R}^d$, where $d$ is a positive integer, as lowercase boldface letters, e.g., $\mathbf{x}$ and $\mathbf{y}$. The $i$th element of $\mathbf{x}$ is denoted by  $x_i$. With this we may write $\vtr{x}$ as
\(
\vtr{x} = \left( x_i \right)_{i \in [d]}
\)
to emphasize its elements, where $[d] = \{1, 2, \ldots, d\}$. For any $\vtr{x} = \left( x_i \right)_{i \in [d]}$, we use the notation $|\vtr{x}|$ to denote the vector $\left( |x_i| \right)_{i \in [d]}$. For any real number $p \geq 1$,  the $\ell_p$ norm of $\mathbf{x} \in \mathbb{R}^d$ is defined as
\[
\norm[p]{\vtr{x}} = \left(\sum_{i= 1}^d |x_i|^p\right)^{1/p}
\]

We define a dataset $X$ as an $n$-tuple each element of which is a vector from $\mathbb{R}^d$. Two datasets $X_1$ and $X_2$ are called \emph{neighbouring datasets} if they differ in one token. The $f$-differential privacy framework is based on the hypothesis testing interpretation of differential privacy. 
First, we give an introduction to this notion, then complete it with the special case of $\epsilon$-Gaussian differential privacy ($\epsilon$-GDP) used in the paper. 

\subsubsection{\texorpdfstring{$f$-Differential Privacy}{f-Differential Privacy}}
\label{app:differential-privacy}
Given the output of a mechanism $\mathcal{M}$, the goal is to distinguish between two competing hypotheses: the underlying data set being $X_1$ or $X_2$. Let $Q_1$ and $Q_2$ denote the probability distributions of $\mathcal{M}(D_1)$ and $\mathcal{M}(D_2)$, respectively. Given any rejection rule $0 \le \phi \le 1$, the type-I and type-II errors are defined as follows ~\cite{dong2022gaussian}: $\alpha_\phi = \mathbb{E}_{Q_1}[\phi]$ and  $\beta_\phi = 1-  \mathbb{E}_{Q_2}[\phi]$.

\begin{definition}[Trade-off function~\cite{dong2022gaussian}]
\label{def:T}
For any two probability distributions $Q_1$ and $Q_2$ on the same
space, the \emph{trade-off function} $T(Q_1, Q_2) : [0, 1] \to [0, 1]$ is defined by
\(
T(Q_1, Q_2)(\alpha) = \inf \{\beta_\phi : \alpha_\phi \leq \alpha\}
\)
for all $\alpha \in [0, 1]$, where the infimum is taken over all (measurable) rejection rules.
\end{definition}

A trade-off function gives the minimum achievable type-II error at any given level of type-I error. For a function to be a trade-off function, it must satisfy the following conditions.

\begin{proposition}[\cite{dong2022gaussian}]
A function $f : [0, 1] \to [0, 1]$ is a trade-off function if and only if $f$ is convex, continuous and non-increasing, and $f(x) \leq 1 - x$ for all $x \in [0, 1]$.
\end{proposition}

Abusing notation, let $\mathcal{M}(X)$ denote the distribution of a mechanism $\mathcal{M}$ when given a dataset $X$ as input.

\begin{definition}[$f$-Differential Privacy ($f$-DP)~\cite{dong2022gaussian}]
Let $f$ be a trade-off function. A mechanism $\mathcal{M}$ is said to
be \emph{$f$-label differentially private} if $T(\mathcal{M}(X_1), \mathcal{M}(X_2)) \geq f$ for all neighbouring data sets $X_1$ and $X_2$.
\end{definition}
For $\epsilon \geq 0$, let $G_\epsilon$ denote the trade-off function between two normal distributions $\mathcal{N}(0, 1)$ and $\mathcal{N}(0, \epsilon)$, i.e., \(G_\epsilon = T(\mathcal{N}(0, 1), \mathcal{N}(0, \epsilon))\). Then
\begin{definition}[$\epsilon$-Gaussian Differential Privacy ($\epsilon$-GDP)~\cite{dong2022gaussian}]
\label{def:gdp}
A mechanism $\mathcal{M}$ is said to satisfy $\epsilon$-Gaussian Differential Privacy ($\epsilon$-GDP) if it is $G_\epsilon$-DP.  That is,
\(
T(\mathcal{M}(X_1), \mathcal{M}(X_2)) \geq G_\epsilon
\)
for all neighbouring data sets $X_1$ and $X_2$.
\end{definition}

\begin{definition}[Global $\ell_2$-Sensitivity]
\label{def:global-sensitivity}
Let $f : \mathbb{R}^d \rightarrow \mathbb{R}^d$ be a function.\footnote{From here onwards, we use the symbol $f$ to denote different functions, as we no longer need to refer to $f$-differential privacy.}  Then its global $\ell_2$-sensitivity, denoted $\Delta f$, is defined as
\(
\Delta f = \max_{\vtr{x}, \vtr{y} \in \mathbb{R}^d} \norm{f(\vtr{x}) - f(\vtr{y})}  
\)
\end{definition}
Our definition of global sensitivity may seem different from the usual definition of global sensitivity in literature which is defined over neighboring datasets (e.g.,~\cite{dp-book}). In our case, the dataset is a set of prompts of $n$ tokens. We say that two prompts are neighboring prompts if they differ in one token. This implies that the $d$-dimenstional embedding of the original token is replaced with that of another token. Therefore, we simply define global sensitivity as the maximum difference over all possible $d$-dimensional token embeddings. Thus, we can talk about applying the $\epsilon$-GDP mechanism token-by-token, instead of defining it over the entire prompt of $n$ tokens, as is defined next. 

\begin{proposition}[$\epsilon$-GDP Mechanism~\cite{dong2022gaussian}]
\label{prop:gdp}
The mechanism $f(X) + \mathcal{N}(0, (\Delta f/\epsilon)^2I_d)$  is $\epsilon$-GDP where $\Delta f$ is the global sensitivity of the function $f : \mathbb{R}^d \rightarrow \mathbb{R}^d$ and $I_d$ is the $d \times d$ identity matrix.
\end{proposition}

The notion of $\epsilon$-GDP satisfies sequential composition.

\begin{proposition}[Sequential Composition~\cite{dong2022gaussian}]
\label{prop:seq-par-comp}
The composition of $n$-fold sequential $\epsilon_i$-GDP mechanisms is 
\(
\sqrt{\epsilon_1^2 + \cdots + \epsilon_n^2}\text{-GDP}
\) 
\end{proposition}

In particular, applying the $\epsilon$-GDP mechanism of \cref{prop:gdp} on $n$ tokens in the prompt makes the overall mechanism $\sqrt{n} \epsilon$-GDP. Lastly, $\epsilon$-GDP is also immune to post processing~\cite{dong2022gaussian}. That is, applying a randomized map with an arbitrary range to the output of an $\epsilon$-GDP mechanism is still $\epsilon$-GDP.

\section{Threat Model and System Architecture}
In our computational model, a trained LLM is hosted on a remote server by a service provider which the user can query via prompts (see \cref{fig:architecture}). The user only receives the final output of the model, which is a sequence of tokens generated by the LLM.
We assume that the server owned by the service provider has two execution environments: a trusted execution environment (TEE) and an untrusted execution environment powered by one or more GPUs, which we simply call GPU. 

\subsection{Privacy Goals and Threat model}
We seek to provide two competing privacy goals
\begin{enumerate}
    \item The user should not infer the  weights of the LLM, which are the intellectual property (IP) of the service provider. 
    \item The service provider should not infer the user prompt.
\end{enumerate}
The first goal is readily achieved in our computational model, as it is in prevalent deployments of LLMs, where the LLM is hosted remotely. We consider model stealing attacks \cite{tramer2016stealing}, where the user crafts prompts to gather enough prompt-response pairs to deduce the model's parameters as out-of-scope.

For the second goal we \textbf{assume} that the TEE is trusted in the sense that any information exchanged between the user and the TEE, and any computation done on the TEE cannot be eavesdropped, even by the service provider. This is a standard assumption for TEEs, as they are generally considered safe environments to perform tasks on sensitive data such as face recognition on smartphones.\footnote{See Apple's Face ID technology as an example: \url{https://support.apple.com/en-au/102381}} On the other hand, the GPU is considered honest-but-curious. Any information transmitted to the GPU, and any computation at the GPU will be carried out as specified, however this is visible to the service provider. %

\subsection{Adversarial Capabilities}
Since the service provider hosts the LLM as a service, it knows all its parameters and architecture. Thus, we assume an adversary, an honest-but-curious service provider, who takes this information about its LLM and any information given to the GPU during LLM inference as input and attempts to reconstruct the user prompt. This is known as an embedding inversion attack (see \cref{sec:rel-embed-inv}). 

\subsection{System Architecture}
Since the TEE is trusted, a straightforward way to achieve our privacy goal is to do the entire LLM inference on the TEE without involving the GPU. However, there is a big performance mismatch between inference at the TEE versus the GPU. Most TEEs, such as the Intel TDX, do not provide GPU-based computation. As a result, the CPU-based computation at the TEE performs poorly compared to GPUs, since the computations involved in LLM inference are highly optimized for GPUs. We illustrate this in detail in \cref{sec:current-results}, where we compare LLM inference times between the TEE (intel TDX) and GPU. For up to 800 generated tokens, the GPU is 10x to 27x faster than the TEE.

We are not the first to notice this disparity. Tram\`er and Boneh~\cite{tramer2018slalom} used this as the motivation to split inference of a conventional deep neural network (DNN)  between a TEE and an untrusted external processor. The same observation was used by Xue et al~\cite{xue2025Securing} to motivate split-inference for LLMs. The main idea in~\cite{tramer2018slalom} to split DNN inference, which was later adopted by Xue et al~\cite{xue2025Securing} for LLM inference, is to outsource \emph{linear operations} to the GPU. 

\subsubsection{Offloading Linear Operations}

To make this more precise, let us first define a linear map. Let $V$ and $V'$ be vector spaces over a field $K$. A function $f : V \rightarrow V'$ is a linear map if for all $\vtr{x}, \vtr{y} \in V$ and scalars $a, b \in K$, $f$ satisfies
\(
f(a \vtr{x} + b\vtr{y}) = af(\vtr{x}) + b f(\vtr{y})
\). 
As an example related to our use case, let $\vtr{x} \in \mathbb{R}^d$ be the input to any of the components of the LLM (\cref{fig:transformer}), e.g., the input embedding, where $d$ is the embedding dimension. Let $\bm{\eta} \in \mathbb{R}^d$ be an arbitrary vector, which we call the \emph{noise vector} or the \emph{mask}. Define the vector $\widetilde{\vtr{x}} = \vtr{x} + \bm{\eta}$ as the \emph{masked input}. Let $W$ be a $d' \times d$ matrix of reals, where $d'$ is potentially different from $d$. Then
\(
W\widetilde{\vtr{x}} = W(\vtr{x} + \bm{\eta}) = W\vtr{x} + W\bm{\eta} = \vtr{r}
\)
holds because this is a linear map from the vector space of dimension $d$ to the vector space of dimension $d'$ over the reals. Thus, any operation in the LLM inference that involves a matrix multiplication is a linear operation. Note that the entries of the matrix are known both to the TEE and the GPU courtesy of the service provider, as these are learned weights of the LLM. In particular, they are provided during setup to the TEE through a defined channel in the TEE's code.

Now the TEE generates the mask, masks the input and hands over the masked input $\widehat{\vtr{x}}$ to the GPU. The GPU can multiply it with the matrix $W$ to obtain the quantity above, and return the result $\vtr{r}$ back to the TEE. The TEE can then straightforwardly compute 
\begin{equation}
\label{eq:add-and-then-subtract}
\vtr{r} - W\bm{\eta} =  (W\vtr{x} + W\bm{\eta}) - W\bm{\eta} = W\vtr{x}    
\end{equation}
This \emph{offloading} process speeds up overall computational time compared to computing the product $W\vtr{x}$ at the TEE due to two main reasons. First, matrix multiplication is computationally expensive, and GPUs are highly optimized to perform them. Secondly, since $W$ is known to the TEE as well, the noise vector $\bm{\eta}$ can be generated beforehand, and the quantity $W \bm{\eta}$ can be precomputed, meaning that all the TEE has to do at runtime is perform one vector addition to produce $\widetilde{\vtr{x}}$, and another vector subtraction to extract $W\vtr{x}$ from the result returned by the GPU. Both of these operations are significantly faster than matrix multiplication at runtime. \cref{fig:Methodology} illustrates this flow. 

\begin{figure}
    \centering
   \resizebox{0.8\columnwidth}{!}{%
\input{Figures/methodology}
}
    \caption{Proposed architecture of outsourcing expensive linear operations during LLM inference to the GPU and then removing the noise similar to~\cite{tramer2018slalom}. In our scheme, the input to the GPU is protected via differentially private noise.}    \label{fig:Methodology}
\end{figure}

Since this approach only relies on the linearity of the offloaded operation, it can be applied directly to matrix multiplications and even linear projections with a bias.\footnote{A linear projection $W$ with a bias produces $W\vtr{x} + \vtr{b}$ when given $\vtr{x}$ as an input. Here $\vtr{b}$ is the bias, which is technically called an affine map.} On the other hand, non-linear operations such as softmax, SiLU, and element-wise multiplication are evaluated only on the TEE, because, in general, an additive mask cannot be removed after applying a non-linear function $f$, i.e., 
\(
f(\vtr{x} + \bm{\eta}) - f(\bm{\eta}) \neq f(\vtr{x}) 
\).

\subsubsection{How to Mask the Inputs?}
The approach taken by~\cite{tramer2018slalom} and \cite{xue2025Securing} is to mask the inputs through a stream cipher. To summarize, they encode (quantize) the input $\vtr{x} \in \mathbb{R}^d$ as a vector of elements of a finite field of appropriate size, and then mask it by adding it to a vector each element of which is a uniformly random element of the finite field. This ensures that no information about the input is leaked to the entity receiving the masked input. For their scheme to work, the weights $W$ of outsourced operations also have to be quantized. Another reason to use operations over a finite field is to then be able to verify that the computations done at the untrusted entity are done correctly through Freivald's verification protocol for matrix products~\cite{freivalds1977probabilistic}. However, quantizing the input and then subsequently doing the linear operations over the finite field introduces approximation errors, which impacts the accuracy of the model. For instance, the accuracy of the quantized CNN models drops up to $0.5\%$ in~\cite{tramer2018slalom}, whereas the quantized LLMs in~\cite{xue2025Securing} show an accuracy drop of up to $1.9\%$. In \cref{subsec:similarity-generated-text} we show that the approach from~\cite{tramer2018slalom} applied to our problem introduces substantial changes in the LLM response. 

Our idea is to instead mask the input using differential privacy. This has the advantage that masked operations can be performed in the native computational domain of LLMs, i.e., floating-point arithmetic, potentially maintaining accuracy.\footnote{Although we do run into floating-point error accumulation issues which we discuss in \cref{subsec:fp-issues}.} We further avoid the additional overhead of repeatedly converting the result between floating-point and finite-field representations if only the offloaded part of the computation is quantized. Differential privacy is a suitable solution in our case as we are not interested in verifying the computation done at the GPU which is assumed to be honest-but-curious. However, we need to determine the magnitude of noise to be added so that the untrusted GPU cannot launch an embedding inversion attack to reconstruct the prompt.

\section{Differential Privacy Analysis}
\cref{fig:transformer} shows the components of LLM inference offloaded to the GPU, which are all linear operations. The sub-component ``Down'' is also a linear operation and hence a candidate for offloading. We discuss the reasons to not outsource this component in \cref{app:down-linear-projection} which stem from the analysis of this section and \cref{subsec:fp-issues}. 
Returning to the figure, the offloaded components labelled  $(Q, K, V)$, Gate/Up and Language Model all follow layer normalization. On the other hand, the $O$ projection follows softmax attention, whereas the Down projection is preceded by SiLU activation. The purpose of this section is to upper bound the possible change in the input to each of these components over all possible input vectors. The idea is to check how much the values can change if any token is replaced by an arbitrary token, and hence a complete change in the embedding vector. Since these inputs are outputs of the preceding components, we find the global sensitivity (\cref{def:global-sensitivity}) of the functions computed in these preceding components. This reduces to finding the global sensitivity of layer normalization, the (softmax) attention mechanism and the (SiLU) activation function. We begin this section with necessary notation, prove some results that will be used in our analysis and then give a detailed global sensitivity analysis of these components. We then briefly discuss how the outsourced components can be made differentially private with the global sensitivity thus calculated. We finish the section with an important note on floating-point errors.

\subsection{Preliminaries}
\label{subsec:prelim}
The dot product of two vectors $\mathbf{x}, \mathbf{y} \in \mathbb{R}^d$, is defined as \(
\dprod{\vtr{x}}{\vtr{y}} = \sum_{i = 1}^d x_i y_i
\)
and their entry-wise product is defined as
\(
\eprod{\vtr{x}}{\vtr{y}} = \left( x_iy_i\right)_{i \in [d]}
\). A vector which we will frequently encounter is the entry-wise product of a vector with itself denoted
\begin{equation}
\label{eq:entry-wise-square}
\osqr{\vtr{x}} = \eprod{\vtr{x}}{\vtr{x}} = \left( x^2_i\right)_{i \in [d]}
\end{equation}
We shall call this the \emph{entry-wise square} of a vector.
The Cauchy-Schwarz inequality states that 
\(
\left| \dprod{\vtr{x}}{\vtr{y}} \right| \leq \norm{\vtr{x}} \norm{\vtr{y}} 
\).
Let $\vtr{x} \in \mathbb{R}^d$ and let $W$ be a matrix of appropriate dimensions. Then 
\begin{equation}
\label{eq:matrix-norm}
    \norm{W\vtr{x}} \leq \norm{W} \norm{\vtr{x}}
\end{equation}
where we define $\norm{W}$ as the spectral norm of $W$, i.e., its largest singular value which can be obtained through its singular value decomposition. This inequality is a standard result; see, for example~\cite{axler2024linear}.

\subsection{Some Useful Results}
In this section, we prove a few results and introduce a few notions that will be helpful in proving theorems in the ensuing section.

\begin{proposition}
\label{prop:z-standard-norm}
Let $\mathbf{x} \in \mathbb{R}^d$. Let 
\(
    \sigma_{\mathbf{x}} = \sqrt{ \frac{1}{d} \sum_{i = 1}^d (x_i - c)^2}
\)
where $c \in \mathbb{R}$. 
Define
\[
\widehat{\mathbf{x}} = \frac{\mathbf{x} - c \cdot \bm{1}}{\sigma_\mathbf{x}}
\]
where $\bm{1}$ is the $d$-element vector each element of which is 1. Then $\norm{\widehat{\mathbf{x}}} = \sqrt{d}$.
\end{proposition}
\begin{proof}
We have 
\begin{align*}
    \norm{\vtr{x}}^2 &= \sum_{i = 1}^d \left( \frac{x_i - c}{\sigma_\mathbf{x}}\right)^2 = \frac{1}{\sigma^2_\vtr{x}} \sum_{i = 1}^d \left(  x_i - c \right)^2 = \frac{d}{\sum_{i = 1}^d \left(  x_i - c \right)^2} \sum_{i = 1}^d \left(  x_i - c \right)^2 = d
\end{align*}
\end{proof}

\begin{definition}
\label{def:normalization}
Let $\vtr{x} \in \mathbb{R}^d$, and let $\gamma > 0$ be a real number. The \emph{z-score normalization} of  $\vtr{x}$ is 
\[
\widehat{\vtr{x}} = \frac{\vtr{x} - \mu_\vtr{x}\cdot \bm{1}}{\sqrt{\sigma_\vtr{x}^2 + \gamma}}
\]
where
\(
    \mu_\mathbf{x} = \frac{1}{d}\sum_{i = 1}^d x_i 
\)
and 
\(
\sigma_{\mathbf{x}}^2 = \frac{1}{d} \sum_{i = 1}^d (x_i - \mu_\mathbf{x})^2
\). 
The \emph{root mean square (RMS) normalization} of $\vtr{x}$ is the z-score normalization of $\vtr{x}$ with $\mu_\vtr{x} = 0$. We say that $\widehat{\vtr{x}}$ is the normalized version of $\vtr{x}$ if it is either z-score or RMS normalized.
\end{definition}

\begin{proposition}
\label{prop:standard-norm-bound}
Let $\vtr{x} \in \mathbb{R}^d$ and let $\widehat{\vtr{x}}$ be its normalized version. Then $\norm{\widehat{\vtr{x}}} \leq \sqrt{d}$.
\end{proposition}
\begin{proof}
Since $\widehat{\vtr{x}}$ is normalized we have
\[
\widehat{\vtr{x}} = \frac{\vtr{x} - \mu_\vtr{x}\cdot \bm{1}}{\sqrt{\sigma_\vtr{x}^2 + \gamma}}
\]
where $\mu_\vtr{x} = \frac{1}{d}\sum_{i = 1}^d x_i$ if it is z-score normalized and $\mu_\vtr{x} = 0$ otherwise. We have
\(
\norm{\widehat{\vtr{x}}} = \left( \sum_{i = 1}^d \widehat{x}_i^2 \right)^\frac{1}{2}.
\)
Consider the $i$th summand 
\[
\widehat{x}_i^2 = \left(\frac{x_i - \mu_\vtr{x}}{\sqrt{\sigma_\vtr{x}^2 + \gamma}}\right)^2 = \frac{(x_i - \mu_\vtr{x})^2}{\sigma_\vtr{x}^2 + \gamma} \leq \frac{(x_i - \mu_\vtr{x})^2}{\sigma_\vtr{x}^2}  = \left( \frac{x_i - \mu_\vtr{x}}{\sigma_\vtr{x}}\right)^2
\]
Thus, from \cref{prop:z-standard-norm} either by setting $c = \mu_\vtr{x} = \frac{1}{d} \sum_{i = 1}^d x_i$ or by setting $c = \mu_\vtr{x} = 0$ we have 
\[
\norm{\widehat{\vtr{x}}} \leq  \left\lVert {\frac{\vtr{x} - \mu_\vtr{x}\cdot \bm{1}}{ \sigma_\vtr{x} }} \right\lVert_2= \sqrt{d}.
\]
\end{proof}

\begin{proposition}
\label{prop:two-plus-two-equals-four}
Let $\vtr{x} \in \mathbb{R}^d$, and let $\osqr{\vtr{x}}$ be its entry-wise square. Then
\(
\norm{\osqr{\vtr{x}}} \leq \norm{\vtr{x}}^2
\).
\end{proposition}
\begin{proof}
We have
$
\vtr{x} = \left( x_i\right)_{i \in [d]}
$
and from Eq.~\ref{eq:entry-wise-square}, by definition
$
\osqr{\vtr{x}} = \left( x^2_i\right)_{i \in [d]}
$

Therefore
\[
\norm{\osqr{\vtr{x}}}^2 = \sum_{i = 1}^d (x_i^2)^2 = \sum_{i = 1}^d x_i^4 
\]
Now,
\begin{align*}
    (\norm{{\vtr{x}}}^2)^2 &=  \left(\sum_{i = 1}^d x_i^2\right)^2 = \left(\sum_{i = 1}^d x_i^2\right) \left(\sum_{i = 1}^d x_i^2\right) = \sum_{i = 1}^d \sum_{j = 1}^d x_i^2 x_j^2 = \sum_{i = 1}^d x_i^2 x_i^2 + \sum_{i = 1}^d \sum_{j = 1, j\ne i}^d x_i^2 x_j^2 = \sum_{i = 1}^d x_i^4 + \sum_{i = 1}^d \sum_{j = 1, j\ne i}^d x_i^2 x_j^2  \\
    &\geq \sum_{i = 1}^d x_i^4 = \norm{\osqr{\vtr{x}}}^2,
\end{align*}
as the omitted term is the sum of non-negative numbers. 
\end{proof}

\begin{proposition}
\label{prop:entrywise-prod-inequality-n}
Let $m \geq 2$ be an integer. Let $\vtr{x}_1, \vtr{x}_2, \ldots ,\vtr{x}_m \in \mathbb{R}^d$, where 
\(
\vtr{x}_i = \left( x_{ij}\right)_{j \in [d]}
\)
for $i \in [m]$. For $i \in [m]$, let $\osqr{\vtr{x}_i}$ denote the entry-wise square of $\vtr{x}_i$. Then
\[
\norm{\vtr{x}_1 \odot \vtr{x}_2 \odot \cdots \odot \vtr{x}_m} \leq \sqrt{\norm{\osqr{\vtr{x}_1}}} \sqrt{\norm{\osqr{\vtr{x}_2}}}\cdots \sqrt{\norm{\osqr{\vtr{x}_m}}} 
\]
\end{proposition}
\begin{proof}
We prove this by induction. Consider the base case with $m = 2$. Then note that 
\(
\norm{\vtr{x}_1 \odot \vtr{x}_2}^2 = \sum_{j = 1}^d x_{1j}^2 x_{2j}^2 =
\dprod{\osqr{\vtr{x}_1}}{\osqr{\vtr{x}_2}}. 
\)
Furthermore, $\dprod{\osqr{\vtr{x}_1}}{\osqr{\vtr{x}_2}}  \geq 0$, as it is the dot product of two vectors all elements of which are non-negative. Thus from the Cauchy-Schwarz inequality we get 
\(
\dprod{\osqr{\vtr{x}_1}}{\osqr{\vtr{x}_2}}  \leq \norm{\osqr{\vtr{x}_1}} \cdot \norm{\osqr{\vtr{x}_2}}.
\)
This immediately gives us
\[
\norm{\vtr{x}_1 \odot \vtr{x}_2} \leq \sqrt{\norm{\osqr{\vtr{x}_1}}} \sqrt{\norm{\osqr{\vtr{x}_2}}}
\]
which proves the base case. Next assume that the statement is true for $m = k$, and consider 
\(
\norm{\vtr{x}_1 \odot \vtr{x}_2 \odot \cdots \odot \vtr{x}_{k+1}}.  
\)
Now $\vtr{x}_1 \odot \vtr{x}_2 \odot \cdots \odot \vtr{x}_{k}$ is a vector in $\mathbb{R}^d$. Denote this by $\vtr{y}$. Then through the base case
\[
\norm{\vtr{x}_1 \odot \vtr{x}_2 \odot \cdots \odot \vtr{x}_{k+1}} = \norm{\vtr{y} \odot \vtr{x}_{k+1}} \leq \sqrt{\norm{\osqr{\vtr{y}}}} \cdot \sqrt{\norm{\osqr{\vtr{x}_{k+1}}}}
\]
From \cref{prop:two-plus-two-equals-four}, we have 
\[
\sqrt{\norm{\osqr{\vtr{y}}}} \leq \sqrt{\norm{\vtr{y}}^2} = \norm{\vtr{y}}
\]
Thus,
\begin{align*}
    \norm{\vtr{x}_1 \odot \vtr{x}_2 \odot \cdots \odot \vtr{x}_{k+1}} &\leq  \norm{\vtr{y}} \cdot \sqrt{\norm{\osqr{\vtr{x}_{k+1}}}} = \norm{\vtr{x}_1 \odot \vtr{x}_2 \odot \cdots \odot \vtr{x}_{k}} \cdot \sqrt{\norm{\osqr{\vtr{x}_{k+1}}}}\\
    &\leq \sqrt{\norm{\osqr{\vtr{x}_1}}} \sqrt{\norm{\osqr{\vtr{x}_2}}}\cdots \sqrt{\norm{\osqr{\vtr{x}_k}}} \sqrt{\norm{\osqr{\vtr{x}_{k+1}}}}
\end{align*}
which follows from the induction hypothesis. This proves the result.
\end{proof}

\begin{proposition}
\label{prop:entrywise-prod-inequality}
Let $\vtr{x}, \vtr{y} \in \mathbb{R}^d$ and let $\widehat{\vtr{y}}$ be the normalized version of $\vtr{y}$. Then
\(
\norm{\eprod{\vtr{x}}{\widehat{\vtr{y}}}} \leq \sqrt{d} \cdot \norm[4]{\vtr{x}}
\)
\end{proposition}
\begin{proof}
From \cref{prop:entrywise-prod-inequality-n}, we have
$\norm{\eprod{\vtr{x}}{\widehat{\vtr{y}}}} \leq  \sqrt{\norm{\osqr{\vtr{x}}}} \cdot \sqrt{\left\lVert{\osqr{\widehat{\vtr{y}}}}\right\rVert_2} $
where $\osqr{\vtr{x}}$ and $\osqr{\widehat{\vtr{y}}}$ are the entry-wise squares of $\vtr{x}$ and $\widehat{y}$, respectively. 
From \cref{prop:two-plus-two-equals-four} we have 
$
\left\lVert{\osqr{\widehat{\vtr{y}}}}\right\rVert_2 \leq \left\lVert{\widehat{{\vtr{y}}}}\right\rVert_2^2
$
and from \cref{prop:standard-norm-bound} we have 
$
\left\lVert{\widehat{{\vtr{y}}}}\right\rVert_2^2 \leq d
$. 
Thus,
$
\sqrt{\left\lVert{\osqr{\widehat{\vtr{y}}}}\right\rVert_2}  \leq \sqrt{d}
$.

Finally, note that 
\begin{align*}
  ({\norm{\osqr{\vtr{x}}}})^{1/2} &= \left(\left( (x_1^2)^2 + (x_2^2)^2 + \cdots + (x_d^2)^2 \right)^{1/2} \right)^{1/2} = \left( x_1^4 + x_2^4 + \cdots + x_d^4 \right)^{1/4} = \norm[4]{\vtr{x}}
\end{align*}
Combining these we get
\begin{equation}
\label{eq:alpha-x-entrywise-prod}
\norm{\eprod{\vtr{x}}{\widehat{\vtr{y}}}}  \leq \sqrt{d} \cdot \norm[4]{\vtr{x}}
\end{equation}    
\end{proof}

\subsection{Global Sensitivity Analysis}
\label{subsec:gs-analysis}

We now analyze global sensitivity (\cref{def:global-sensitivity}) of the functions preceding the offloaded components. In our analysis, we use the same dimension $d$ for all operations. As discussed in \cref{subsec:llm-components}, some components use different dimensions, e.g., $d_\text{key}$. However, this is merely an implementation issue, and our sensitivity analysis can be used by replacing $d$ with the exact dimension applicable to the operation.  

\subsubsection{Initial Tokens}
As a warm-up we look at the initial layer, even though it is not offloaded. 
At initialization, the input prompt is converted into tokens, and then each token is represented by a token embedding $\vtr{x} \in \mathbb{R}^d$. The token embeddings are from the vocabulary $\mathsf{vocab} \subseteq \mathbb{R}^d$ of the LLM. Global sensitivity, in this case for an identity function, can be found experimentally as 
\(
\max_{\vtr{x}, \vtr{y} \in \mathsf{vocab}} \norm{\vtr{x} - \vtr{y}} 
\).
Note that we do not need to span the entire space $\mathbb{R}^d$, since the token embeddings are only from $\mathsf{vocab}$. 
\subsubsection{After Layer Normalization}
The layer normalization function $\mathsf{LN} : \mathbb{R}^d \rightarrow \mathbb{R}^d$ as defined in~\cite{ba2016layer} is
\begin{equation}
\label{eq:ln}
\mathsf{LN}(\vtr{x}; \bm{\alpha}, \bm{\beta}, \gamma) = \eprod{\bm\alpha}{\widehat{\vtr{x}}} + \bm{\beta}
\end{equation}
where
$\bm{\alpha}, \bm{\beta} \in \mathbb{R}^d$ are learned parameters, and $\widehat{\vtr{x}}$ is the z-score normalization of $\vtr{x} \in \mathbb{R}^d$.\footnote{The parameter $\gamma$ in the z-score normalization as defined in \cref{def:normalization} does not appear in~\cite{ba2016layer} but it is used in actual implementations for numerical stability to ensure the denominator never becomes 0.} In some models such as the new Llama models, $\widehat{\vtr{x}}$ is RMS normalized, known as the RMSNorm~\cite{zhang2019Root}. The benefit and use of the RMSNorm in Llama models is discussed in~\cite{gupta2026Geometric}. The following result holds for both standardization methods. We will mostly skip any reference to the parameters $\bm{\alpha}, \bm{\beta}$ and $\gamma$, and simply write $\mathsf{LN}(\vtr{x})$ instead of $\mathsf{LN}(\vtr{x}; \bm{\alpha}, \bm{\beta}, \gamma)$.  
\begin{theorem}
\label{theorem:gs-layer-norm}
Let $\mathsf{LN}$ be as defined in Eq.~\eqref{eq:ln}. Then its global sensitivity $\Delta \mathsf{LN}$ satisfies
\(
\Delta \mathsf{LN} \leq 2 \sqrt{d} \cdot \norm[4]{\bm{\alpha}}.
\)
\end{theorem}
\begin{proof}
For any $\vtr{x}, \vtr{y} \in \mathbb{R}^d$, 
\begin{align}
    \norm{\mathsf{LN}(\vtr{x}) - \mathsf{LN}(\vtr{y})} &= \norm{\eprod{\bm\alpha}{\widehat{\vtr{x}}} + \bm{\beta} - \eprod{\bm\alpha}{\widehat{\vtr{y}}} - \bm{\beta} } \nonumber = \norm{\eprod{\bm{\alpha}}{(\widehat{\vtr{x}} - \widehat{\vtr{y}})}} \nonumber = \norm{\eprod{\bm{\alpha}}{\widehat{\vtr{x}}} - \eprod{\bm{\alpha}}{\widehat{\vtr{y}})}} \nonumber\\
    &\leq \norm{\eprod{\bm{\alpha}}{\widehat{\vtr{x}}}} + \norm{\eprod{\bm{\alpha}}{\widehat{\vtr{y}}} }\nonumber
\end{align}
where in the last step we have used the triangle inequality for the Euclidean norm. From \cref{prop:entrywise-prod-inequality} we get 
\[
\norm{\eprod{\bm{\alpha}}{\widehat{\vtr{x}}}} \leq \sqrt{d} \cdot \norm[4]{\bm{\alpha}} \quad \text{ and } \quad \norm{\eprod{\bm{\alpha}}{\widehat{\vtr{y}}}} \leq \sqrt{d} \cdot \norm[4]{\bm{\alpha}}
\]
Putting these inequalities in the above inequality, we get
\[
\norm{\mathsf{LN}(\vtr{x}) - \mathsf{LN}(\vtr{y})} \leq \norm{\eprod{\bm{\alpha}}{\widehat{\vtr{x}}}} + \norm{\eprod{\bm{\alpha}}{\widehat{\vtr{y}}} } \leq 2\sqrt{d} \cdot \norm[4]{\bm{\alpha}}
\]
Since $\vtr{x}, \vtr{y} \in \mathbb{R}^d$ were arbitrary, we conclude
\(
\Delta\mathsf{LN} \leq 2\sqrt{d} \cdot \norm[4]{\bm{\alpha}}.
\)
\end{proof}

\subsubsection{After the Attention Mechanism}
The attention mechanism can be defined as follows. Let $W_\text{qry}$, $W_{\text{key}}$ and $W_\text{val}$ be $d \times d$ matrices. Suppose we have $n$ vectors $\vtr{x}_1, \vtr{x}_2, \ldots, \vtr{x}_n \in \mathbb{R}^d$. For $1 \leq i \leq n$, define \(
\vtr{q}_i = W_\text{qry} \cdot \vtr{x}_i,  \vtr{k}_i = W_\text{key} \cdot \vtr{x}_i\) and \(\vtr{v}_i = W_\text{val} \cdot \vtr{x}_i
\)
which are again vectors from $\mathbb{R}^d$. Let $\vtr{q} = W_\text{qry} \vtr{x}$ for some $\vtr{x} \in \mathbb{R}^d$. For $1 \leq i \leq n$ define
\[
p_{i}(\vtr{x}) = \frac{\exp(\dprod{\vtr{q}}{\vtr{k}_i}/\sqrt{d})}{\sum_{j = 1}^n \exp(\dprod{\vtr{q}}{\vtr{k}_j}/\sqrt{d})} = \text{softmax}\left( \frac{\dprod{\vtr{q}}{\vtr{k}_i}}{\sqrt{d}}\right)
\]
Note that $\sum_{i = 1}^n p_i(\vtr{x}) = 1$. With this, the attention mechanism for any $\vtr{x} \in \mathbb{R}^d$ is defined as
\begin{equation}
\label{eq:att-mech}
\mathsf{att}(\vtr{x}) =  \sum_{i = 1}^n p_i(\vtr{x}) \vtr{v}_i
\end{equation}
The input to the attention function is in fact layer normalized. Thus, we can write it as 

\begin{align}
\mathsf{att}(\vtr{x}) &=  \sum_{i = 1}^n p_i(\widehat{\vtr{x}}) W_{\text{val}} \cdot \mathsf{LN}(\vtr{x}_i) \nonumber = \sum_{i = 1}^n p_i(\widehat{\vtr{x}}) W_{\text{val}} \cdot(\eprod{\bm{\alpha}}{\widehat{\vtr{x}}_i + \bm{\beta}}) \nonumber\\
&= \sum_{i = 1}^n p_i(\widehat{\vtr{x}}) W_{\text{val}} \cdot(\eprod{\bm{\alpha}}{\widehat{\vtr{x}}_i)} + \sum_{i = 1}^n p_i(\widehat{\vtr{x}}) W_{\text{val}} \cdot \bm{\beta} \nonumber = \sum_{i = 1}^n p_i(\widehat{\vtr{x}}) W_{\text{val}} \cdot(\eprod{\bm{\alpha}}{\widehat{\vtr{x}}_i)} + \left(\sum_{i = 1}^n p_i(\widehat{\vtr{x}})\right)  W_{\text{val}} \cdot \bm{\beta} \nonumber \\
&= W_{\text{val}} \cdot \bm{\beta} + \sum_{i = 1}^n p_i(\widehat{\vtr{x}}) W_{\text{val}} \cdot(\eprod{\bm{\alpha}}{\widehat{\vtr{x}}_i)} \label{eq:att-mech-normalized}
\end{align}
where $\widehat{\vtr{x}}$ and $\widehat{\vtr{x}}_i$ are normalized versions of $\vtr{x}$ and $\vtr{x}_i$ respectively.
We are interested in finding $\Delta \mathsf{att}$. 

\begin{theorem}
\label{theorem:gs-att-mech}
Let $\mathsf{att}$ be as defined in Eq.~\eqref{eq:att-mech-normalized}. Then its global sensitivity $\Delta \mathsf{att}$ satisfies
\(
\Delta \mathsf{att} \leq 2\sqrt{d} \cdot \norm{W_\text{val}}\cdot \norm[4]{\bm{\alpha}}.
\)
\end{theorem}
\begin{proof}
By definition, for any $\vtr{x}, \vtr{y} \in \mathbb{R}^d$  
\[
\mathsf{att}(\vtr{x}) = W_{\text{val}} \cdot \bm{\beta} + \sum_{i = 1}^n p_i(\widehat{\vtr{x}}) W_{\text{val}} \cdot(\eprod{\bm{\alpha}}{\widehat{\vtr{x}}_i }) 
\qquad\text{and}\qquad
\mathsf{att}(\vtr{y}) = W_{\text{val}} \cdot \bm{\beta} + \sum_{i = 1}^n p'_i(\widehat{\vtr{y}}) W_{\text{val}} \cdot(\eprod{\bm{\alpha}}{\widehat{\vtr{x}}'_i })
\]
Note that we use different variables $p'_i$ and $\widehat{\vtr{x}}'_i$ for $\mathsf{att}(\vtr{y})$ since replacing $\vtr{x}$ with $\vtr{y}$ changes the softmax \emph{probabilities} $p_i(\widehat{\vtr{x}})$ and the sequence of vectors $\widehat{\vtr{x}}_i$ due to one vector being different. Now 
\begin{align}
    \norm{\mathsf{att}(\vtr{x}) - \mathsf{att}(\vtr{y})} 
    &= \left\lVert \sum_{i = 1}^n p_i(\widehat{\vtr{x}}) W_{\text{val}} \cdot(\eprod{\bm{\alpha}}{\widehat{\vtr{x}}_i }) - \sum_{i = 1}^n p'_i(\widehat{\vtr{y}}) W_{\text{val}} \cdot(\eprod{\bm{\alpha}}{\widehat{\vtr{x}}'_i }) \right\rVert_2 \nonumber\\
    &\leq \left\lVert \sum_{i = 1}^n p_i(\widehat{\vtr{x}}) W_{\text{val}} \cdot(\eprod{\bm{\alpha}}{\widehat{\vtr{x}}_i }) \right\rVert_2 + \left\lVert  \sum_{i = 1}^n p'_i(\widehat{\vtr{y}}) W_{\text{val}} \cdot(\eprod{\bm{\alpha}}{\widehat{\vtr{x}}'_i }) \right\rVert_2 \nonumber \\
    &\leq \sum_{i = 1}^n \left\lVert  p_i(\widehat{\vtr{x}}) W_{\text{val}} \cdot(\eprod{\bm{\alpha}}{\widehat{\vtr{x}}_i }) \right\rVert_2 + \sum_{i = 1}^n\left\lVert   p'_i(\widehat{\vtr{y}}) W_{\text{val}} \cdot(\eprod{\bm{\alpha}}{\widehat{\vtr{x}}'_i }) \right\rVert_2 \nonumber \\
    &= \sum_{i = 1}^n p_i(\widehat{\vtr{x}}) \left\lVert   W_{\text{val}} \cdot(\eprod{\bm{\alpha}}{\widehat{\vtr{x}}_i }) \right\rVert_2 + \sum_{i = 1}^n p'_i(\widehat{\vtr{y}}) \left\lVert    W_{\text{val}} \cdot(\eprod{\bm{\alpha}}{\widehat{\vtr{x}}'_i }) \right\rVert_2 \label{eq:gs-att-mech-int}
\end{align}
where we have used standard properties of the Euclidean norm. 
Now let us examine first of the two norms above in isolation. Let $1 \leq i \leq n$ be fixed. From Eq.~\eqref{eq:matrix-norm}, We have
\begin{align*}
\left\lVert   W_{\text{val}} \cdot(\eprod{\bm{\alpha}}{\widehat{\vtr{x}}_i }) \right\rVert_2 &\leq \norm{W_\text{val}} \cdot \norm{\eprod{\bm{\alpha}}{\widehat{\vtr{x}}_i }} \\
&\leq \norm{W_\text{val}} \cdot \sqrt{d} \cdot \norm[4]{\bm{\alpha}}
\end{align*}
where the second inequality follows from \cref{prop:standard-norm-bound}. Likewise, for any $i \in \{1, \ldots, n\}$, we have 
\[
\left\lVert   W_{\text{val}} \cdot(\eprod{\bm{\alpha}}{\widehat{\vtr{x}}'_i }) \right\rVert_2 \leq \norm{W_\text{val}} \cdot \sqrt{d} \cdot \norm[4]{\bm{\alpha}}
\]
Thus, Eq.~\eqref{eq:gs-att-mech-int} becomes
\begin{align*}
    \norm{\mathsf{att}(\vtr{x}) - \mathsf{att}(\vtr{y})}  &\leq \sum_{i = 1}^n p_i(\widehat{\vtr{x}}) \norm{W_\text{val}} \cdot \sqrt{d} \cdot \norm[4]{\bm{\alpha}} + \sum_{i = 1}^n p'_i(\widehat{\vtr{y}}) \norm{W_\text{val}} \cdot \sqrt{d} \cdot \norm[4]{\bm{\alpha}} \\
    &= \left(\sum_{i = 1}^n p_i(\widehat{\vtr{x}}) \right)\norm{W_\text{val}} \cdot \sqrt{d} \cdot \norm[4]{\bm{\alpha}} + \left(\sum_{i = 1}^n p'_i(\widehat{\vtr{y}})\right) \norm{W_\text{val}} \cdot \sqrt{d} \cdot \norm[4]{\bm{\alpha}}\\
    &= 2\sqrt{d} \cdot \norm{W_\text{val}}\cdot \norm[4]{\bm{\alpha}}
\end{align*}
Since $\vtr{x}, \vtr{y} \in \mathbb{R}^d$ were arbitrary, we conclude
\(
\Delta \mathsf{att} \leq 2\sqrt{d} \cdot \norm{W_\text{val}}\cdot \norm[4]{\bm{\alpha}}
\)
\end{proof}

\subsubsection{After the Activation Function}
The activation function in the feed-forward neural network can be defined as follows. Let $W_\text{gate}$ and $W_\text{up}$ be two $d \times d$ matrices. Let $g(\vtr{x}) = W_\text{gate} \cdot \vtr{x}$ for $\vtr{x} \in \mathbb{R}^d$. For any $x \in \mathbb{R}$ the \emph{sigmoid} function is defined as 
\begin{equation}
\label{eq:def:sigmoid}
\mathsf{sig}(x) = \frac{1}{1 + e^{-x}}    
\end{equation}
Note that for any $x \in \mathbb{R}$, $\mathsf{sig}(x) \in (0, 1)$. 
We extend this to a vector $\vtr{x} \in \mathbb{R}^d$ by defining $\mathsf{sig}(\vtr{x})$ to be the $d$-dimensional vector whose $i$th element is $\mathsf{sig}(x_i)$. With this the activation function on $\vtr{x}$ is defined as
\begin{equation*}
\mathsf{act}(\vtr{x}) = W_\text{up} \cdot \vtr{x} \odot W_\text{gate} \cdot \vtr{x} \odot  \mathsf{sig}(g(\vtr{x}))
\end{equation*}
The term $W_\text{gate} \cdot \vtr{x} \odot  \mathsf{sig}(g(\vtr{x}))$ is known as the sigmoid linear unit or SiLU. Just like the attention mechanism, this function is applied on the normalized version of $\vtr{x}$. Therefore, we can write the activation function as
\begin{equation}
\label{eq:activation-function}
\mathsf{act}(\vtr{x}) = W_\text{up} \cdot \mathsf{LN}(\vtr{x}) \odot W_\text{gate} \cdot \mathsf{LN}(\vtr{x}) \odot  \mathsf{sig}(g(\vtr{x}))
\end{equation}
where we have re-defined $g(\vtr{x)}$ as $W_\text{gate} \cdot \mathsf{LN}(\vtr{x})$.
With this, we have the following result for $\Delta \mathsf{act}$. 
\begin{theorem}
\label{theorem:gs-act-function}
Let $\mathsf{act}$ be as defined in Eq.~\eqref{eq:activation-function}. Then its global sensitivity $\Delta \mathsf{act} $ satisfies
\[
\Delta \mathsf{act} \leq 2  \sqrt{d}\cdot \left(\sqrt{d} \cdot \norm[4]{\bm{\alpha}} + \norm{\bm{\beta}}\right)^2
 \cdot \norm{W_\text{up}} \cdot \norm{W_\text{gate}}
\]
\end{theorem}
\begin{proof}
For any $\vtr{x}, \vtr{y} \in \mathbb{R}^d$, we have
\begin{align}
    \norm{\mathsf{act}(\vtr{x}) - \mathsf{act}(\vtr{y})} &= \left\lVert W_\text{up} \cdot \mathsf{LN}(\vtr{x}) \odot W_\text{gate} \cdot \mathsf{LN}(\vtr{x}) \odot  \mathsf{sig}(g(\vtr{x})) \right. \nonumber \left. - W_\text{up} \cdot \mathsf{LN}(\vtr{y}) \odot W_\text{gate} \cdot \mathsf{LN}(\vtr{y}) \odot  \mathsf{sig}(g(\vtr{y}))\right\rVert_2 \nonumber\\
    &\leq \norm{W_\text{up} \cdot \mathsf{LN}(\vtr{x}) \odot W_\text{gate} \cdot \mathsf{LN}(\vtr{x}) \odot  \mathsf{sig}(g(\vtr{x}))} + \norm{W_\text{up} \cdot \mathsf{LN}(\vtr{y}) \odot W_\text{gate} \cdot \mathsf{LN}(\vtr{y}) \odot  \mathsf{sig}(g(\vtr{y}))} \label{eq:gs-act-int}
\end{align}
which follows from the triangle inequality of the Euclidean norm. Consider the first of the two norms above
\begin{align}
    \quad \norm{W_\text{up} \cdot \mathsf{LN}(\vtr{x}) \odot W_\text{gate} \cdot \mathsf{LN}(\vtr{x}) \odot  \mathsf{sig}(g(\vtr{x}))}  &\leq \sqrt{\norm{(W_\text{up} \cdot \mathsf{LN}(\vtr{x}))^{\odot 2}}} \cdot \sqrt{\norm{(W_\text{gate} \cdot \mathsf{LN}(\vtr{x}))^{\odot 2}}} \cdot \sqrt{\norm{(\mathsf{sig}(g(\vtr{x})))^{\odot 2}}} \nonumber\\
    &\leq 
    \norm{W_\text{up} \cdot \mathsf{LN}(\vtr{x})} \cdot \norm{W_\text{gate} \cdot \mathsf{LN}(\vtr{x})} \cdot \norm{\mathsf{sig}(g(\vtr{x}))} \nonumber\\
    &\leq \norm{W_\text{up}} \cdot \norm{\mathsf{LN}(\vtr{x})} \cdot \norm{W_\text{gate}} \cdot \norm{\mathsf{LN}(\vtr{x})} \cdot \norm{\mathsf{sig}(g(\vtr{x}))} \nonumber\\
    &= \norm{W_\text{up}} \cdot \norm{W_\text{gate}} \cdot \norm{\mathsf{LN}(\vtr{x})}^2 \cdot \norm{\mathsf{sig}(g(\vtr{x}))} \label{eq:act-norm-int}
\end{align}
The first inequality follows from \cref{prop:entrywise-prod-inequality-n}, the second from \cref{prop:two-plus-two-equals-four}, and the third from Eq.~\eqref{eq:matrix-norm}. Now from Eq.~\eqref{eq:ln}, 
\begin{align*}
    \norm{\mathsf{LN}(\vtr{x})} = \norm{\bm{\alpha} \odot \widehat{\vtr{x}} + \bm{\beta}} \leq \norm{\bm{\alpha} \odot \widehat{\vtr{x}}} + \norm{\bm{\beta}} \leq \sqrt{d} \cdot \norm[4]{\bm{\alpha}} + \norm{\bm{\beta}} 
\end{align*}
where the last inequality follows from \cref{prop:entrywise-prod-inequality}. 
Thus, 
\begin{equation}
\label{eq:ln-norm-bound}
     \norm{\mathsf{LN}(\vtr{x})}^2 \leq \left(\sqrt{d} \cdot \norm[4]{\bm{\alpha}} + \norm{\bm{\beta}}\right)^2
\end{equation}
For the other term, we have
\begin{align*}
    \norm{\mathsf{sig}(g(\vtr{x}))}^2 &= \sum_{i = 1}^d (\mathsf{sig}((g(\vtr{x}))_i) )^2 \leq \sum_{i = 1}^d 1^2 = d 
\end{align*}
Therefore $ \norm{\mathsf{sig}(g(\vtr{x}))} \leq \sqrt{d}$. Thus, Eq.~\eqref{eq:act-norm-int} becomes
\begin{align*}
    \norm{W_\text{up} \cdot \mathsf{LN}(\vtr{x}) \odot W_\text{gate} \cdot \mathsf{LN}(\vtr{x}) \odot  \mathsf{sig}(g(\vtr{x}))} &\leq \norm{W_\text{up}} \cdot \norm{W_\text{gate}} \cdot \norm{\mathsf{LN}(\vtr{x})}^2 \cdot \norm{\mathsf{sig}(g(\vtr{x}))}  \\
    &\leq \sqrt{d}\cdot \left(\sqrt{d} \cdot \norm[4]{\bm{\alpha}} + \norm{\bm{\beta}}\right)^2
 \cdot \norm{W_\text{up}} \cdot \norm{W_\text{gate}}
\end{align*}
Exactly the same way, we have
\begin{align*}
  \norm{W_\text{up} \cdot \mathsf{LN}(\vtr{y}) \odot W_\text{gate} \cdot \mathsf{LN}(\vtr{y}) \odot  \mathsf{sig}(g(\vtr{y}))} \leq \sqrt{d}\cdot \left(\sqrt{d} \cdot \norm[4]{\bm{\alpha}} + \norm{\bm{\beta}}\right)^2
 \cdot \norm{W_\text{up}} \cdot \norm{W_\text{gate}}  
\end{align*}
Therefore, Eq.~\eqref{eq:gs-act-int} becomes
\begin{align*}
\norm{\mathsf{act}(\vtr{x}) - \mathsf{act}(\vtr{y})} \leq 2  \sqrt{d}\cdot \left(\sqrt{d} \cdot \norm[4]{\bm{\alpha}} + \norm{\bm{\beta}}\right)^2
 \cdot \norm{W_\text{up}} \cdot \norm{W_\text{gate}}
\end{align*}
Since $\vtr{x}, \vtr{y} \in \mathbb{R}^d$ were arbitrary, we conclude
\[
\Delta \mathsf{act} \leq  2  \sqrt{d}\cdot \left(\sqrt{d} \cdot \norm[4]{\bm{\alpha}} + \norm{\bm{\beta}}\right)^2
 \cdot \norm{W_\text{up}} \cdot \norm{W_\text{gate}}
\]
\end{proof}
\subsection{Making the Process Differentially Private}
\label{subsec:dp-proof}
Given the $\ell_2$-sensitivity $\Delta f$ of the output $\vtr{x}$ of the component $f$, which serves as the input to the component being offloaded to the GPU, we first sample noise $\bm{\eta} \sim \mathcal{N}(\mathbf{0}, ({\Delta f}/{\epsilon})^2 I_d)$, where $I_d$ is the $d \times d$ identity matrix. We then offload the vector $\vtr{x} + \bm{\eta}$ to the GPU. For a prompt of $n$ tokens, with a \emph{single} offloaded operation, this mechanism satisfies $\sqrt{n} \epsilon$-Gaussian differential privacy according to \cref{prop:gdp,prop:seq-par-comp}. The steps taken by the GPU maintain this guarantee due to the post-processing property of GDP. However, note that a given prompt is offloaded multiple times: once for each offloaded component, once for each transformer layer, and then once for each generated token. Thus, to be precise, the overall privacy guarantee should include all offloaded instances. However, we stick to the privacy guarantee of $\sqrt{n} \epsilon$-GDP and even $\epsilon$-GDP (i.e., keeping the privacy parameter constant, instead of scaling it with the prompt size), as otherwise an overwhelming amount of noise needs to be added, and moreover, our prompt reconstruction attacks (\cref{subsec:prompt-reconstruction-attack}) show negligible success rate with this privacy level. 

\subsection{Floating Point Issues}
\label{subsec:fp-issues}
The noise addition and cancellation procedure in our scheme 
as shown in Eq.~\eqref{eq:add-and-then-subtract} works perfectly for real numbers. However, in floating-point arithmetic, this may result in a phenomenon known as \emph{catastrophic cancellation} (see for example~\cite[\S 4.2]{muller2018handbook}). Given the quantities $W \vtr{x}$ and $W(\vtr{x} + \bm{\eta}) - W\bm{\eta}$ we are interested in knowing the difference between the two, when the arithmetic is done over floating point numbers. The $i$th element of $W \vtr{x}$ is given as $\dprod{\vtr{w}_i}{\vtr{x}}$, where $\vtr{w}_i$ is the $i$th row of $W$. We thus consider a generic element of $W \vtr{x}$, denoted $\dprod{\vtr{w}}{\vtr{x}}$, as the dot product of the appropriate row of $W$. The corresponding value through offloading is $\dprod{\vtr{w}}{\vtr{x}+ \bm{\eta}} - \dprod{\vtr{w}}{\bm{\eta}}$. 

To make this precise, we will use the notation $\float{}$ to denote the variant of any function that uses two floating point numbers and outputs a floating point number as an answer. For example if $x$ and $y$ are floating point numbers, then $\float(x + y)$ denotes the operation that adds two floating point numbers and rounds them to the nearest floating point number. Under this notation, we are interested in the quantity
\begin{equation}
\label{eq:fp-error}
\left| \float({\dprod{\vtr{w}}{\vtr{x}}}) - \float( \float(\dprod{\vtr{w}}{\vtr{x}+ \bm{\eta}}) - \float(\dprod{\vtr{w}}{\bm{\eta}}))\right|
\end{equation}
Note that the outer subtraction is not over floating point numbers.

\subsubsection{Background On Floating Point Arithmetic}  

We give a brief background on floating point arithmetic summarized from~\cite{muller2018handbook}. 
We use the binary32 number format as the foundation. A number $x$ in this format is written as 
\(
x = M \cdot 2^{e - p + 1},
\)
where $p = 24$ is the precision, $e$ is an integer satisfying $e_{\min} \leq e \leq e_{\max}$ and $M$ is an integer satisfying $2^{p - 1} \leq |M| < 2^p$. To be precise this is called the normalized representation. For binary32 numbers we have $e_{\min} = -126$ and $e_{\max} = 127$. The \emph{unit round-off} of radix-2 precision-$p$ floating point system in the round-to-nearest mode is defined as
\begin{equation}
\label{eq:ulp}
    u = \frac{1}{2}2^{1-p} = 2^{-p}
\end{equation}
This gives rise to the following standard model of floating-point arithmetic. For all floating-point numbers $a$, $b$ such
that $a * b$ does not underflow nor overflow, we have
\begin{equation}
\label{eq:std-model}
\float{(a*b)} = (a * b)(1 + \delta), \;|\delta| < u    
\end{equation}
where $*$ is one of the elementary arithmetic operations: $+, -, \times, \div$. Roughly, an underflow (respectively, overflow) occurs when the result is smaller (respectively, larger) than the smallest (respectively, largest) possible floating point number in the given number format.

The following definition is from \cite[\S 3.1]{higham2002accuracy} (see also \cite[\S 5.1]{muller2018handbook})
\begin{definition}
\label{def:u-and-theta}
Let $\delta_i$ satisfy $|\delta_i| \le u$ for $1 \le i \le d$, and assume that
\(
du < 1
\). Then
\(
\prod_{i=1}^{d} (1+\delta_i)^{\pm 1} = 1+\theta_d,
\)
where
\(
|\theta_d| \le \frac{du}{1-du} =: \gamma_d.
\)
\qed
\end{definition}
We will make use of the following result from~\cite[\S 3.1]{higham2002accuracy}:
\begin{equation}
\label{eq:dot-product-bound}
\left| \dprod{\vtr{w}}{\vtr{x}} - \float({\dprod{\vtr{w}}{\vtr{x}}}) \right|
\leq \gamma_d \dprod{|\vtr{w}|}{|\vtr{x}|}
\end{equation}

\subsubsection{Floating-Point Error Bound}

Our main result is as follows
\begin{theorem}
\label{theorem:fp-bound}
If no underflow or overflow occurs then
\begin{align*}
&\left| \float({\dprod{\vtr{w}}{\vtr{x}}}) - \float( \float(\dprod{\vtr{w}}{\vtr{x}+ \bm{\eta}}) - \float(\dprod{\vtr{w}}{\bm{\eta}}))\right| \leq \gamma_d \dprod{|\vtr{w}|}{|\vtr{x}| + |\vtr{x} + \bm{\eta}| +  |\bm{\eta}|} + u \dprod{|\vtr{w}|}{|\vtr{x}|} + \mathcal{O}(du^2)
\end{align*}
\end{theorem}
\begin{proof}
Let us define, 
\[
A = \float({\dprod{\vtr{w}}{\vtr{x}}}), B = \float(\dprod{\vtr{w}}{\vtr{x}+ \bm{\eta}}), C = \float(\dprod{\vtr{w}}{\bm{\eta}})
\]
With this notation, we seek to find the bound on
\[
\left| A - \float(B - C) \right|
\]
From one application of the standard model (Eq.~\eqref{eq:std-model}) we have
\begin{align}
\left| A - \float(B - C) \right| &\leq 
\left| A - (B - C)(1 + \delta) \right| \nonumber = 
\left| A - \dprod{\vtr{w}}{\vtr{x}} + \dprod{\vtr{w}}{\vtr{x}} - (B - C)(1 + \delta) \right| \nonumber\\
&\leq \left| A - \dprod{\vtr{w}}{\vtr{x}} \right| + \left| (B - C)(1 + \delta) - \dprod{\vtr{w}}{\vtr{x}} \right| \label{eq:a-and-b-and-c}
\end{align}
where the last step follows from the triangle inequality. From Eq.~\eqref{eq:dot-product-bound}, we have
\begin{equation}
\label{eq:a-bound}
\left| A - \dprod{\vtr{w}}{\vtr{x}} \right| = \left| \float({\dprod{\vtr{w}}{\vtr{x}}}) - \dprod{\vtr{w}}{\vtr{x}} \right| \leq \gamma_d \dprod{|\vtr{w}|}{|\vtr{x}|}
\end{equation}
We expand the second term in Eq.~\eqref{eq:a-and-b-and-c} as follows
\begin{align}
\left| (B - C)(1 + \delta) - \dprod{\vtr{w}}{\vtr{x}} \right| &= \left| B - C +\delta(B-C) - \dprod{\vtr{w}}{\vtr{x}} \right| \nonumber\\
&\leq \left| B - C  - \dprod{\vtr{w}}{\vtr{x}} \right| + \left| \delta(B-C)\right| = \left| B - C  - \dprod{\vtr{w}}{\vtr{x}} + \dprod{\vtr{w}}{\bm{\eta}} - \dprod{\vtr{w}}{\bm{\eta}}\right| + \left| \delta(B-C)\right| 
\nonumber\\
&\leq \left| B - \dprod{\vtr{w}}{\vtr{x}} - \dprod{\vtr{w}}{\bm{\eta}}\right| + \left| C  - \dprod{\vtr{w}}{\bm{\eta}}\right| + \left| \delta(B-C)\right| \nonumber\\
&= \left| B - \dprod{\vtr{w}}{\vtr{x} + \bm{\eta}}\right| + \left| C  - \dprod{\vtr{w}}{\bm{\eta}}\right| + \left| \delta(B-C)\right|
\label{eq:b-and-c}
\end{align}
with multiple applications of the triangle inequality. 
From Eq.~\ref{eq:std-model}, we have
\begin{equation}
\label{eq:b-bound}
\left| B - \dprod{\vtr{w}}{\vtr{x} + \bm{\eta}}\right| = \left|\float(\dprod{\vtr{w}}{\vtr{x} + \bm{\eta}}) - \dprod{\vtr{w}}{\vtr{x} + \bm{\eta}}\right| \leq \gamma_d \dprod{|\vtr{w}|}{|\vtr{x} + \bm{\eta}|} 
\end{equation}
and
\begin{equation}
\label{eq:c-bound}
\left| C - \dprod{\vtr{w}}{\bm{\eta}}\right| = \left|\float(\dprod{\vtr{w}}{\bm{\eta}}) - \dprod{\vtr{w}}{\bm{\eta}}\right| \leq \gamma_d \dprod{|\vtr{w}|}{|\bm{\eta}|} 
\end{equation}
The remaining term in Eq.~\eqref{eq:b-and-c} can be expanded as follows, noting that $|\delta| \leq u$ from \cref{def:u-and-theta},
\begin{align}
 \left| \delta(B-C)\right| &\leq |\delta|\left| B - C\right| \nonumber\\
 &\leq u |B-C| \nonumber = u |B - C + \dprod{\vtr{w}}{\vtr{x} + \bm{\eta}} - \dprod{\vtr{w}}{\vtr{x} + \bm{\eta}}| \nonumber\\
 & \leq u\left| B - \dprod{\vtr{w}}{\vtr{x} + \bm{\eta}} \right| + u \left| C - \dprod{\vtr{w}}{\bm{\eta}} \right| + u|\dprod{\vtr{w}}{\vtr{x}}|\nonumber\\
 & \leq u\left| B - \dprod{\vtr{w}}{\vtr{x} + \bm{\eta}} \right| + u \left| C - \dprod{\vtr{w}}{\bm{\eta}} \right| + u\dprod{|\vtr{w}|}{|\vtr{x}|
 }|\nonumber
\end{align}
which again follows from repeated use of the triangle inequality. From Eqs.\eqref{eq:b-bound} and \eqref{eq:c-bound}, we have
\begin{align}
 \left| \delta(B-C)\right| &\leq u\gamma_d \dprod{|\vtr{w}|}{|\vtr{x} + \bm{\eta}|} + u \gamma_d \dprod{|\vtr{w}|}{|\bm{\eta}|} + u\dprod{|\vtr{w}|}{|\vtr{x}|} \nonumber\\
 &= u\gamma_d \dprod{|\vtr{w}|}{|\vtr{x} + \bm{\eta}| + |\bm{\eta}|} + u\dprod{|\vtr{w}|}{|\vtr{x}|} \label{eq:delta-b-and-c}
\end{align}
By using the fact that 
\[
\frac{1}{1 - x} = 1 + x + x^2 + \cdots
\]
and denoting $\dprod{|\vtr{w}|}{|\vtr{x} + \bm{\eta}| + |\bm{\eta}|} = c $ as a constant, we see that the first term in Eq.~\eqref{eq:delta-b-and-c} becomes
\begin{align*}
    u\gamma_d \dprod{|\vtr{w}|}{|\vtr{x} + \bm{\eta}| + |\bm{\eta}|} &= c u \gamma_d = cu \frac{du}{1 -du} = cdu^2 \frac{1}{1 - du} = cdu^2 (1 + du + d^2u^2 + \cdots) = cdu^2 + cd^2u^3 + cd^3u^4 + \cdots \\
    &= \mathcal{O}(du^2)
\end{align*}
where we have used the definition of $\gamma_d$ from \cref{def:u-and-theta}. Thus, Eq.~\eqref{eq:delta-b-and-c} becomes
\begin{equation}
\label{eq:delta-b-and-c-2}
\left| \delta(B-C)\right| \leq u\dprod{|\vtr{w}|}{|\vtr{x}|}  + \mathcal{O}(du^2)
\end{equation}
Putting the results from Eqs.~\eqref{eq:b-bound}, \eqref{eq:c-bound} and \eqref{eq:delta-b-and-c-2} in Eq~\eqref{eq:b-and-c}, we get 
\begin{align}
\left| (B - C)(1 + \delta) - \dprod{\vtr{w}}{\vtr{x}} \right| \leq \gamma_d \dprod{|\vtr{w}|}{|\vtr{x} + \bm{\eta}|} + \gamma_d \dprod{|\vtr{w}|}{| \bm{\eta}|} + u\dprod{|\vtr{w}|}{|\vtr{x}|}  + \mathcal{O}(du^2) \label{eq:b-and-c-2}
\end{align}
Finally, putting the results from Eqs.~\eqref{eq:a-bound} and \eqref{eq:b-and-c-2} in Eq.~\ref{eq:a-and-b-and-c}, we get
\begin{align*}
    \left| A - \float(B - C) \right| &\leq \gamma_d \dprod{|\vtr{w}|}{|\vtr{x}|} + \gamma_d \dprod{|\vtr{w}|}{|\vtr{x} + \bm{\eta}|} + \gamma_d \dprod{|\vtr{w}|}{| \bm{\eta}|} + u\dprod{|\vtr{w}|}{|\vtr{x}|}  + \mathcal{O}(du^2) \\
&= \gamma_d \dprod{|\vtr{w}|}{|\vtr{x}| + |\vtr{x} + \bm{\eta}| +  |\bm{\eta}|} + u \dprod{|\vtr{w}|}{|\vtr{x}|} + \mathcal{O}(du^2)
\end{align*}
as desired.
\end{proof}

How may we use the result of ~\cref{theorem:fp-bound}? Since $\vtr{w}$, $d$ and $u$ are constants, the error bound depends on $\vtr{x}$ and $\bm{\eta}$. The latter is a Gaussian random variable with mean $\vtr{0}$ and scale $\frac{\Delta f}{\epsilon} I_d$. Furthermore, the range of values of $\vtr{x}$ are known, e.g., after layer normalization. Thus, through a Monte Carlo simulation we can determine the error bound (ignoring the $\mathcal{O}(du^2)$ term) for a given value of $\epsilon$. For this value of $\epsilon$, we can determine whether the overall error is acceptable in the sense that it will not change the recovered dot product by too much. Through our simulations, we found the following bound to be closer to the true error
\begin{align}
&\left| \float({\dprod{\vtr{w}}{\vtr{x}}}) - \float( \float(\dprod{\vtr{w}}{\vtr{x}+ \bm{\eta}}) - \float(\dprod{\vtr{w}}{\bm{\eta}}))\right|  \leq \gamma_{\lceil \log_2 d \rceil + 1} \dprod{|\vtr{w}|}{|\vtr{x}| + |\vtr{x} + \bm{\eta}| + |\bm{\eta}|} + u \dprod{|\vtr{w}|}{|\vtr{x}|}\label{eq:fp-log-bound}   
\end{align}
which is based on Higham's estimate~\cite[\S 3.1, p. 64]{higham2002accuracy} for computing the dot product through pairwise summation. We will therefore use this bound in our experiments. Note that the only quantity in the bound above that depends on the specific floating-point number format is $u$ which, as defined in Eq.~\eqref{eq:ulp}, depends on the precision $p$ of the format. Thus, this bound is transferrable to any floating-point number format such as binary32, binary16 (half-precision floating point format) or brain floating point (bfloat16)~\cite{raschka2023acceleratinglargelanguagemodelswithm}.%

\section{Implementation}
\label{sec:implementation}

We use Intel Trust Domain Extensions (TDX) as our CPU-based TEE~\cite{chengTDXDemystified2024}. The Intel TDX architecture allows instantiating hardware-isolated virtual machines, which we call TDX guests (sometimes referred to as Trust Domains), whose memory and execution state are protected from the hypervisor, other TDX guests, and any other software running on the host. We implement the split-inference architecture on a single physical machine by partitioning LLM inference between a TDX guest (virtual machine) and a honest-but-curious GPU running on the host. In addition to the guest's private memory, Intel TDX supports sharing a memory zone with the host to enable guest-host communication. The detailed host and TDX configuration used for our implementation is provided in \suppmat.

\subsection{Offloaded Operations}
The components of the LLM offloaded to the GPU are illustrated in \cref{fig:transformer}. All these components involve matrix multiplications, specifically, with the matrices: $W_\text{qry}, W_\text{key}, W_\text{val}, W_\text{oh}, W_\text{gate}, W_\text{up}$, and $W_\text{lm}$ previously presented in \cref{subsec:llm-components}. We did not offload the down linear projection due to the reasons described below. We shall denote a generic matrix belonging to this set by $W$. 
The masked input and the result are shared between the TEE and the GPU through a shared memory. This is illustrated in \cref{fig:Methodology} for a generic matrix $W$.

One way to resolve this issue is to perform any computation from this component in 64-bit floating-point arithmetic. The reason why this would work is evident from Eq.~\eqref{eq:fp-log-bound}. In the IEEE 754 binary64 format, also called double-precision floating-point, we have precision $p = 53$, meaning that $u = 2^{-53}$ from Eq.~\eqref{eq:ulp}. Consequently, $\gamma_d$ from \cref{def:u-and-theta} is $\gamma_d \approx d \cdot 2^{-53}$, assuming $du \ll 0$, as should be the case. On the other hand, in the IEEE 754 binary32 format, also called single-precision floating-point, we have $p = 24$, meaning that $u = 2^{-24}$, and $\gamma_d \approx d \cdot {2^{-24}}$ with the same assumption. The ratio of this to the former is approximately $2^{29} \approx 5.38 \times 10^8$. Thus the bound in Eq.~\eqref{eq:fp-log-bound} in the 64-bit floating-point is many orders of magnitude smaller than smaller precision formats, and hence it can solve the floating-point error amplification issue with noise correction. However, using 64-bit precision for the Down component significantly increases the noise bank size, as each stored value requires more memory. The higher precision also increases memory consumption, processing time, and the amount of data transferred through shared memory. Consequently, outsourcing this operation provides diminishing returns. We therefore retain the Down component inside the TEE.

\subsection{Precomputing Noise}
\label{sub:precompute-noise}

Instead of generating the noise $\bm{\eta}$ and the product $W\bm{\eta}$ at inference time, we generate them in advance. This is essentially a time-memory tradeoff, and it is the same strategy used in~\cite{tramer2018slalom} and \cite{xue2025Securing}. In particular, $W\bm{\eta}$ is itelf a matrix product, and computing it at runtime would defeat the purpose of offloading expensive operations to the GPU. 
Note that we cannot offload this multiplication to the GPU as it reveals the noise vector, violating the differential privacy guarantee.
For any offloaded component, equivalently, for any matrix $W$ associated with the offloaded component, we call $\bm{\eta}$ the \emph{mask} and $W\bm{\eta}$ the \emph{correction}, as this is the quantity the TEE has to subtract from the result obtained from the GPU. We generate the masks and corrections in advance for each offloaded component and store them in a \emph{noise bank}. Each noise entry, consisting of the mask and its corresponding correction, is consumed only once, preventing reuse of the same mask across different calls to these components. Additional details on the design and organization of the noise bank are provided in \suppmat.

\subsection{Down Projection Exclusion}\label{app:down-linear-projection}
The ``Down'' linear projection described in \cref{sec:ffnl} and depicted in \cref{fig:transformer} is also a linear projection involving the matrix multiplication $W_\text{down}$. The input to this component is the output of the activation function (SiLU). \cref{theorem:gs-act-function} shows the bound on the sensitivity of the activation function, and therefore the scale of noise that needs to be added to the input $\vtr{x}$ to the Down component. In our experiments with \llama{}, we found that offloading this component caused significant storage cost for the noise bank and floating-point errors (\cref{subsec:fp-issues}) after noise correction. This is due to two main reasons. First the input to this component has dimension (also called intermediate dimension) equal to $8,192$ ($d_\text{ff}$), which is significantly larger than the model's embedding dimension (hidden dimension) $d = 3,072$. Furthermore, substituting this dimension size in the global sensitivity bound of \cref{theorem:gs-act-function} in $d$ and the norms of the matrices $W_\text{gate}$ and $W_\text{up}$ yields a value much higher than the global sensitivity bounds for other components. The net result is that the noise bank has to store noise of larger scale and size, as well as significant amplification of floating-point errors after noise correction. The latter means that there is a significant mismatch between $W \vtr{x}$ and $W(\vtr{x} + \bm{\eta}) - W\bm{\eta}$, where the latter is calculated via offloading. This error is then carried over to other components of the LLM, resulting in a drastically different output. 

\subsection{Inference Models}
\label{subsec:inference-models}

We evaluate our approach using two instruction-tuned LLMs: \texttt{Llama-\allowbreak3B} and \texttt{Qwen3-4B-Instruct}, which we simply call \qwen{} from now onwards. %
Both are decoder-only transformer models with RMSNorm and SiLU-based MLP (Feed-Forward Network) blocks, but they differ in their internal dimensions and attention configurations (see \suppmat \ for additional architectural details).

\section{Results}
\label{sec:current-results}
We first compare inference times, followed by similarity of generated text with the baseline, and finish by analyzing the prompt reconstruction attack.

\subsection{Inference Time}

We begin by defining the configurations used to compare inference times, as shown in \cref{tab:configurations}. The CPU and GPU configurations serve as baselines for full-model inference on the host CPU and GPU, respectively, with no security guarantee. The TDX configuration measures full-model inference inside the protected TDX guest. The \tconfig{} configuration represents vanilla split inference without masking and correction. The remaining two configurations, \epssharp{} and \epsroot{}, implement our protected split-inference design using different noise scales. In \epssharp{}, a constant $\epsilon$ is used per token (e.g., $\epsilon = 1$), whereas \epsroot{} uses $\epsilon = 1/\sqrt{n}$ per token for a prompt containing $n$ tokens.

\begin{table}
\centering
\caption{The various configurations to compare LLM inference times.}
\begin{tabular}{r|p{0.8\columnwidth}}
CPU & Entire LLM inference is performed by the host CPU.\\
TDX & Entire LLM inference is performed within the TDX guest, using the host CPU. However, the computation and the guest's memory are protected by Intel TDX. \\
GPU & Entire LLM inference is performed by the host GPU.\\
\tconfig{} & Split-inference between the TDX guest and host GPU without masking and correction.\\
\epssharp{} & Split-inference between the TDX guest and host GPU with constant noise scale $\epsilon$. Over $n$ tokens of the prompt, this gives $\sqrt{n}\epsilon$-Gaussian differential privacy (see \cref{subsec:dp-proof}). \\
\epsroot{} & Split-inference between the TDX guest and host GPU with noise of scale $\epsilon = 1/\sqrt{n}$. Over $n$ tokes, this gives $\sqrt{n}\epsilon = 1$-Gaussian differential privacy.
\end{tabular}

\label{tab:configurations}
\end{table}

For each configuration, we use 45 prompts selected from the Databricks Dolly dataset~\cite{conover2023dolly}%
, with an average input length of about 100 tokens. Databricks Dolly is a human-generated instruction-response dataset designed for instruction-following LLMs.  We use a batch size of 15, resulting in three batches per configuration. Each configuration is executed once over the same set of prompts, and we report the average inference time per prompt. Batching allows multiple prompts to be processed in a single run and reduces the number of separate data transfers and GPU calls. We evaluate four output-token limits of 100, 200, 400, and 800 tokens.

\begin{figure*}
    \centering
    \begin{subfigure}{\linewidth}
        \includegraphics[width=\textwidth]{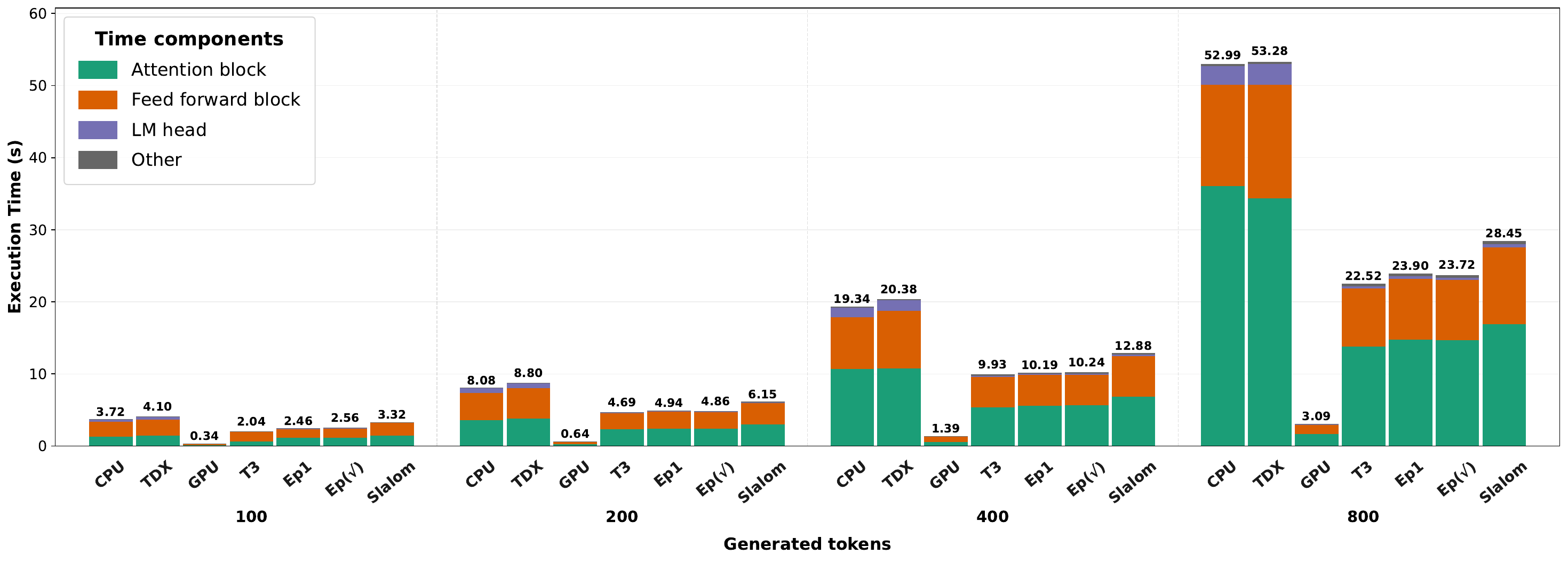}
        \caption{\llama{}}
        \label{fig:runtime-comparison}
    \end{subfigure}
    \begin{subfigure}{\linewidth}
        \includegraphics[width=\textwidth]{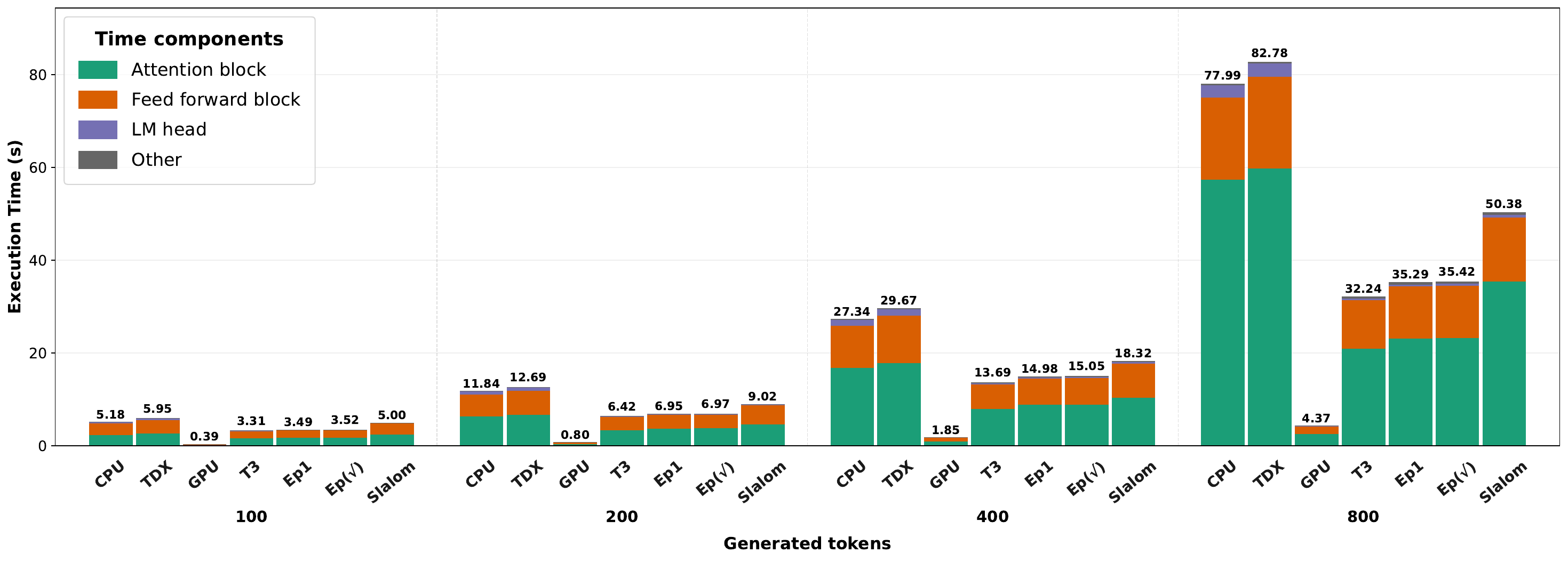}
        \caption{\qwen{}}
        \label{fig:runtime-comparison-qwen}
    \end{subfigure}
\caption{Execution time comparison for \llama{} (top) and \qwen{} (bottom) across CPU-only, full-TDX, direct-GPU, \textsc{T3} (no-noise split), the noised configurations \textsc{EP1} and \textsc{EP}$(1/\sqrt{n})$, and an adapted Slalom baseline using single-precision (FP32) outsourced computation.}
\label{fig:AllModelsPerfs}
\end{figure*}

\cref{fig:runtime-comparison} shows that CPU and TDX inference are respectively 13x and 14x much slower on average than the untrusted GPU, motivating the split execution we propose in this paper to benefit from GPU speedup. 
Indeed, the split execution without noise addition (\tconfig{}) is 2x faster on average than the full TDX execution, with greater speedups as more tokens are generated.
Finally, noise addition only brings a slight overhead, namely a 1.08x and 1.09x overhead with \epsone{} and \epsroot{}, respectively. 
Overall, protecting the privacy of the prompt with our scheme incurs a 8.4x overhead on average, compared to a 14x overhead with a basic full TEE execution.

The times in \cref{fig:runtime-comparison} are divided into times taken by the attention block (linear operations \(Q\), \(K\), \(V\), and \(O\)), the feed-forward block (operations Gate and Up), the LM head, and all other components. As can be seen, most of the performance improvement over the TDX configuration comes by offloading the attention, the feed-forward and LM components. This justifies the need to offload these components to the GPU as it can handle these expensive operations in parallel. 
The remaining operations have minimal impact on the overall time and are hence performed inside the TDX guest (see the legend ``Other'' in \cref{fig:runtime-comparison}).

To examine whether the results generalize to another LLM, we repeat the experiment using \qwen{} in \cref{fig:runtime-comparison-qwen}. The trend is the same as that of \llama{}, except now there is a bigger difference between the TDX and noisy configurations. For example, for 800 generated tokens, the TDX configuration takes almost 83 seconds whereas \epsone{} and \epsroot{} take up to 36 seconds only.

\subsection{Inference Time Comparison with Slalom}
\label{subsec:baseline}
As mentioned before, the main inspiration of our work comes from the closest existing system to ours, Slalom~\cite{tramer2018slalom}, presented in \cref{sec:rel-work-teeLLM}. 
While their solution applies to convolutional neural network, the major difference between Slalom and our work is how they mask the input to the offloaded linear operations. They work in the finite field $\mathbb{Z}_q$ modulo a prime $q$. The $d$-dimensional mask $\bm{\eta}$ is then a vector each element of which is a random element from the field $\mathbb{Z}_q$. The input $\vtr{x} \in \mathbb{R}^d$ is then \emph{quantized} into an integer vector $\vtr{x}_\text{qt} \in \mathbb{Z}_q^d$ of the same dimension via fixed-point representation. The masked input is then $\widetilde{\vtr{x}}_\text{qt} = \vtr{x}_\text{qt} + \bm{\eta} \pmod{q}$. Likewise, the weight matrix $W$ is quantized as $W_\text{qt}$ into a matrix of elements from $\mathbb{Z}_q$. The GPU computes $W_\text{qt}\widetilde{\vtr{x}}_\text{qt} = W_\text{qt}(\vtr{x}_\text{qt} + \bm{\eta}) \pmod{q}$, and sends the result back to the TEE, who then corrects the noise by computing  $W_\text{qt}(\vtr{x}_\text{qt} + \bm{\eta}) - W_\text{qt}\bm{\eta}$. The result can then be de-quantized back to a vector in $\mathbb{R}^d$ (or more precisely a vector of floating-point equivalents) to continue operations in the original domain inside the TEE. It is important to remark that Slalom also checks the integrity of the result of matrix multiplication from the GPU through Freivalds' algorithm~\cite{freivalds1977probabilistic}, hence another reason to work within the field $\mathbb{Z}_q$. However, we omit this part as we do not consider the GPU to be malicious in our threat model.

To compare the inference times and accuracy of our mechanism based on differential privacy against that of Slalom's, based on cryptography (essentially a one-time pad), we adapt Slalom's masking mechanism to the same LLM components outsourced in our system. 
Before sending an input $\vtr{x}$ to the GPU, we quantize it and then mask it using a fresh random vector. The weight matrices at the GPU are also quantized. The GPU performs the matrix computation in the finite field $\mathbb{Z}_q$. The returned result is then corrected using precomputed (quantized) correction as shown above. The recovered result is then de-quantized back to floating point before the remaining operations continue inside the TDX.

Following Slalom's implementation, we use the prime modulus \(q = 2^{23}+2^{21}+7 = 10{,}485{,}767\) and 8-bit quantization for both input tensor and model weights. \cref{fig:runtime-comparison} shows that the run-time through Slalom is slower than our method. Specifically, for 400 generated tokens Slalom is 2 seconds slower, whereas for 800 generated tokens it is almost 5 seconds slower than our scheme.

\subsection{Similarity of Generated Text}
\label{subsec:similarity-generated-text}
To test if our scheme does not degrade the quality of the response from the LLM, we compare the outputs generated by the original and the split-inference models under greedy decoding. Greedy decoding outputs the next token with the highest score (logit) produced by the language model head. 

\subsubsection{Similarity Score} 
We evaluate 100 prompts with input lengths between 100 and 250 tokens, selected from the Databricks Dolly dataset. For each prompt, we generate exactly 200 new tokens using greedy decoding. We evaluate privacy levels $\epsilon\in \{0.5,1,5,10,15\}$ for the \epssharp{} configuration and $\epsilon=1/\sqrt{n}$ for the \epsroot{} configuration, where $n$ denotes the number of tokens in the input prompt. We compare the responses against the GPU configuration as the baseline. 
For the comparison, we use exact output match, token match rate, BLEU, ROUGE-L, and semantic similarity. Exact output match measures the percentage of prompts with identical generated sequences, while token match rate measures the percentage of matching generated tokens. BLEU captures n-gram overlap~\cite{papineni2002bleu}, and ROUGE-L measures longest common subsequence similarity~\cite{lin2004rouge}. Semantic similarity is the cosine similarity between response embeddings from the \texttt{all-mpnet-base-v2} model~\cite{reimers-2019-sentence-bert}.

\cref{tab:similarity-score} shows that \epssharp{} outputs become more stable as $\epsilon$ increases. At $\epsilon=0.5$, 33\% of tokens diverge, while no deviations occur at $\epsilon=10$ or $15$. Token match, BLEU, and ROUGE-L follow the same trend, while semantic similarity remains high, from 99.36\% to 100\%. The stronger \epsroot{} perturbation, with an average $\epsilon$ of 0.0773, reduces exact matches to 10\%, yet semantic similarity remains 95.56\%, showing that the overall meaning is largely preserved.

\begin{table}
\centering
\caption{Output comparison between the original model and the split-noised model under greedy decoding. 
}
\label{tab:similarity-score}
\resizebox{0.7\columnwidth}{!}{%
\begin{tabular}{c c c c c c c}
\toprule
$\epsilon$ &
\textbf{Exact-match} &
\textbf{Token match} &
\textbf{BLEU} &
\textbf{ROUGE-L} &
\textbf{Semantic similarity} \\
\midrule
$1/\sqrt{n}$ & 10.00\%  & 42.55\% & 64.49\% & 69.92\% & 95.56\% \\
0.5 & 67.00\%  & 85.43\% & 91.59\% & 93.26\% & 99.36\% \\
1   & 84.00\%  & 91.49\% & 94.44\% & 95.38\% & 99.67\% \\
5   & 98.00\%   & 99.41\% & 99.57\% & 99.59\% & 99.96\% \\
10  & 100.00\%  & 100.00\% & 100.00\% & 100.00\% & 100.00\% \\
15  & 100.00\%  & 100.00\% & 100.00\% & 100.00\% & 100.00\% \\
\midrule
\textbf{Slalom} & 0.00\%  & 20.22\% & 43.82\% & 52.28\% & 93.11\% \\
\bottomrule
\end{tabular}
}
\end{table}

\paragraph{Comparison with Slalom.} Slalom shows substantially larger output divergence than our configurations. None of its 100 responses match the GPU baseline exactly, with a token match rate of 20.22\%, BLEU of 43.82\%, ROUGE-L of 52.28\%, and semantic similarity of 93.11\%. The advantage of our scheme is that $\epsilon$ in \epssharp{} can be tuned to achieve a desired accuracy level, whereas Slalom's accuracy loss stems from quantization and is not tunable. A full 200-token example comparing the GPU baseline, \epsroot{}, \epsone{}, and Slalom is provided in \cref{tab:similarity-example}.
\begin{table}
\centering
\caption{Example prompts and generated responses for the GPU baseline, \epsroot{}, and Slalom. Each response is limited to 200 generated tokens and therefore ends at the generation cap. Differences from the GPU baseline are shown in {\color{red!80!black}\textbf{red}} for both \epsroot{} and Slalom.}
\label{tab:similarity-example}
\begin{tabularx}{\textwidth}{@{}X@{}}
\toprule
\textbf{Prompt:}
When was the first Reading railway station opened? \emph{Context:} Reading railway station is a major transport hub in Reading, Berkshire, England. It is on the northern edge of the town centre, near the main retail and commercial areas and the River Thames, 36 miles (58 km) from London Paddington. The first Reading station was opened on 30 March 1840 as the temporary western terminus of the original line of the Great Western Railway (GWR). Reading is the ninth-busiest station in the UK outside London and the second busiest interchange station outside London with over 3.8 million passengers changing trains at the station annually.
\\
\midrule
\end{tabularx}

\begin{tabularx}{\textwidth}{@{}XX@{}}
\textbf{GPU} &
\textbf{\epsone{}} 
\\
The first Reading railway station was opened on 30 March 1840. It was the temporary western terminus of the original line of the Great Western Railway (GWR). The station has since undergone several expansions and renovations, but this date marks its initial opening. The station has become a major transport hub in Reading, serving over 3.8 million passengers annually. It is also a key interchange station, with connections to other parts of the UK. The station's location near the River Thames and the town centre makes it a significant transportation hub in the region. The station's history dates back to 1840, and it has played a crucial role in the development of Reading as a major town in Berkshire, England. The station has undergone several expansions and renovations over the years, but its initial opening in 1840 marked the beginning of its growth as a major transport hub. The station's strategic location near the River Thames and the town centre has made it an essential transportation link for the region
&
The first Reading railway station was opened on 30 March 1840. It was the temporary western terminus of the original line of the Great Western Railway (GWR). The station has since undergone several expansions and renovations, but this date marks its initial opening. The station has become a major transport hub in Reading, serving over 3.8 million passengers annually. It is also a key interchange station, with connections to other parts of the UK. The station's location near the River Thames and the town centre makes it a significant transportation hub in the region. The station's history dates back to 1840, and it has played a crucial role in the development of Reading as a major town in Berkshire, England. The station has undergone several expansions and renovations over the years, but its initial opening in 1840 marked the beginning of its growth as a major transport hub. The station's strategic location near the River Thames and the town centre has made it an essential transportation link for the region
\\
\midrule
\end{tabularx}

\begin{tabularx}{\textwidth}{@{}XX@{}}
\textbf{\epsroot{}} &
\textbf{Slalom} 
\\
The first Reading railway station was opened on 30 March 1840. It was the temporary western terminus of the original line of the Great Western Railway (GWR). The station has since undergone several expansions and renovations, but this date marks its initial opening. The station has become a major transport hub in Reading, serving over 3.8 million passengers annually. It is also a {\color{red!80!black}\textbf{significant}} interchange station, with connections to other {\color{red!80!black}\textbf{lines and trains.}} The station's location near the River Thames and the town centre makes it a {\color{red!80!black}\textbf{convenient and accessible point for commuters and travelers.}} The station's history dates back to {\color{red!80!black}\textbf{the 19th century,}} and it has played a {\color{red!80!black}\textbf{vital}} role in the development of Reading as a major {\color{red!80!black}\textbf{transportation hub. The station's current status as the ninth-busiest station in the UK outside London and the second busiest interchange station outside London is a testament to its importance in the region. The station's growth and development over the years have made it an essential part of the transportation infrastructure in}}
&
The first Reading railway station was opened on 30 March 1840. It was the temporary western terminus of the original line of the Great Western Railway (GWR). {\color{red!80!black}\textbf{This was over 183 years ago.}} The station has since undergone several expansions and renovations, but {\color{red!80!black}\textbf{its initial opening date remains 30 March 1840. It's worth noting that the station has been in operation for nearly two centuries, and it continues to be}} a major transport hub in Reading, serving over 3.8 million passengers annually. {\color{red!80!black}\textbf{Its strategic}} location near the River Thames and the town centre makes it {\color{red!80!black}\textbf{an essential transportation link to London and other parts of the country. The station's history and evolution over the years have played a significant}} role in the development of Reading {\color{red!80!black}\textbf{and the surrounding region. It has been an integral part of the town's growth and transformation into the thriving city it is today. The station's rich history and ongoing importance in the transportation network make it a notable landmark in the UK, with}}
\\
\bottomrule
\end{tabularx}
\end{table}

\subsubsection{Floating-Point Errors} The main reason for response divergence in our scheme is the accumulation of floating-point errors during noise correction as discussed in \cref{subsec:fp-issues}. To measure this effect, we compute the numerical error specified in Eq.~\eqref{eq:fp-error}, together with the error bound in Eq.~\eqref{eq:fp-log-bound}, for individual output elements of all outsourced linear operators, namely the \(Q\), \(K\), \(V\), \(O\), Gate, and Up projections, as well as the final LM-head projection.

We use 100 Dolly prompts of 100-250 tokens and capture the inputs to each outsourced operator during inference. We evaluate both \epssharp{} and \epsroot{}, using multiple $\epsilon$ values for \epssharp{}. Each setting covers 169 operator instances: six outsourced projections across 28 transformer layers and the LM head. For each instance, we run 1{,}000 trials by randomly sampling a captured input vector and corresponding weight row and applying Gaussian noise using the operator-specific global sensitivity. This yields 169{,}000 evaluated output elements per $\epsilon$. We report the mean numerical error and computed bound, using $p=24$ for IEEE binary32 precision.

\cref{tab:numerical-error} reports the mean true error and corresponding mean bound, both decreasing as $\epsilon$ increases and masking noise is reduced. The bound is a good indicator of the error, at approximately four times the mean observed error. It is also easier to compute because it avoids end-to-end inference. Given $W$, we can directly sample expected inputs $\vtr{x}$ and generate the mask $\bm{\eta}$. We like to remark that there were some instances where the actual error was higher than the computed bound, as indicated by the last column in \cref{tab:numerical-error}. This is because our bound is based on a particular dot-product computation, while machine learning libraries may sometimes use optimizations not reflected in our simplified calculation. 

\begin{table}[t]
\centering
\caption{Numerical-error results across 
the evaluated \(\epsilon\) values. Bound violation is the fraction of evaluated output elements for which the measured numerical error exceeds the computed bound.}
\label{tab:numerical-error}
\resizebox{0.70\columnwidth}{!}{%
\begin{tabular}{ccccc}
\toprule
\(\boldsymbol{\epsilon}\) &
\textbf{Mean true error} &
\textbf{Mean bound} &
\textbf{Bound/error ratio} &
\textbf{Bound violation (\%)} \\
\midrule

\(1/\sqrt{n}\) &
\(1.52\times10^{-3}\) &
\(6.53\times10^{-3}\) &
\(4.30\times\) &
9.93\% \\

0.5 &
\(2.26\times10^{-4}\) &
\(9.64\times10^{-4}\) &
\(4.27\times\) &
10.03 \\

1 &
\(1.13\times10^{-4}\) &
\(4.83\times10^{-4}\) &
\(4.28\times\) &
9.89 \\

3 &
\(3.75\times10^{-5}\) &
\(1.61\times10^{-4}\) &
\(4.30\times\) &
9.83 \\

5 &
\(2.26\times10^{-5}\) &
\(9.67\times10^{-5}\) &
\(4.29\times\) &
9.70 \\

7 &
\(1.61\times10^{-5}\) &
\(6.91\times10^{-5}\) &
\(4.29\times\) &
9.68 \\

10 &
\(1.12\times10^{-5}\) &
\(4.86\times10^{-5}\) &
\(4.33\times\) &
9.51 \\

15 &
\(7.48\times10^{-6}\) &
\(3.25\times10^{-5}\) &
\(4.35\times\) &
9.25 \\

\bottomrule
\end{tabular}
}
\end{table}

\paragraph{How Much Numerical Error is Too Much?}
To understand how much numerical error the model can tolerate before output quality starts to degrade, we performed an additional controlled error (CE) experiment. We used 100 prompts from the same Dolly dataset and evaluated all prompts at each controlled error level. For each outsourced operator, we first measured the numerical error that is naturally created by the masking and unmasking process under \epsone{}. This error is the difference between the clean projection output under \tconfig{} and the corrected output recovered after masking, and is computed separately for every operator invocation rather than using a single fixed error value. We then preserved the error pattern observed for each operator call and changed only its magnitude. This allowed us to gradually increase the numerical error seen by the model while keeping the structure of the naturally occurring masking error.

\cref{tab:controlled-error} reports the results of this experiment. The \tconfig{} setting represents clean split inference without masking and therefore has zero numerical error. The \epsone{} setting uses the masking mechanism without any adjustment to the resulting error. Its mean absolute error of \(2.06\times10^{-4}\) therefore corresponds to the naturally occurring error produced by the masking and correction process. For the controlled-error settings, the \epsone{} error pattern is preserved for each operator call and its magnitude is adjusted to the mean absolute error levels shown in \cref{tab:controlled-error}, ranging from \(1.00\times10^{-4}\) to \(3.00\times10^{-2}\). These values therefore represent controlled error levels rather than naturally occurring errors. The generated output at each error level is compared with the corresponding \tconfig{} output using semantic similarity. We compute this using the \texttt{all-mpnet-base-v2} Sentence-Transformer model and cosine similarity, reported as a percentage.

The results show that the natural numerical error introduced under \epsone{} has little effect on generation quality. At a mean absolute error of \(2.06\times10^{-4}\), semantic similarity remains at \(99.20\%\). When the mean absolute error is increased to \(1.00\times10^{-3}\), similarity remains at \(94.48\%\), while an error of \(1.00\times10^{-2}\) reduces it to \(82.99\%\). A much larger drop is observed at higher error levels, with semantic similarity decreasing to \(65.36\%\) at \(1.40\times10^{-2}\), \(43.04\%\) at \(1.60\times10^{-2}\), and \(9.91\%\) at \(2.00\times10^{-2}\). These results show how generation quality changes as the numerical error increases and provide a practical reference for interpreting the error values reported in \cref{tab:numerical-error}. More specifically, since the error changes with the value of $\epsilon$ in the \epssharp{} configuration, we can check the error bound from Eq.~\eqref{eq:fp-log-bound}. If the bound is above $~10^{-4}$, the table suggests that the corresponding computations will degrade generation quality. It is the feature of our scheme that this error can be controlled by choosing a suitable value of $\epsilon$ balancing privacy and accuracy.

\begin{table}[t]
\centering
\caption{Effect of controlled numerical error on generation quality. CE denotes controlled error, where the \epsone{} error is scaled to the specified mean absolute error.}
\label{tab:controlled-error}
\resizebox{0.85\columnwidth}{!}{%
\begin{tabular}{c c c | c c c}
\toprule
\textbf{Setting} &
\shortstack{\textbf{Mean absolute error}} &
\shortstack{\textbf{Semantic similarity}} &
\textbf{Setting} &
\shortstack{\textbf{Mean absolute error}} &
\shortstack{\textbf{Semantic similarity}} \\
\midrule

\tconfig{} & $0$ & $100.00\%$ &
CE & $1.20\times10^{-2}$ & $74.83\%$ \\

\epsone{} & $2.06\times10^{-4}$ & $99.20\%$ &
CE & $1.40\times10^{-2}$ & $65.36\%$ \\

CE & $1.00\times10^{-4}$ & $99.49\%$ &
CE & $1.60\times10^{-2}$ & $43.04\%$ \\

CE & $3.00\times10^{-4}$ & $97.98\%$ &
CE & $1.80\times10^{-2}$ & $20.69\%$ \\

CE & $1.00\times10^{-3}$ & $94.48\%$ &
CE & $2.00\times10^{-2}$ & $9.91\%$ \\

CE & $3.00\times10^{-3}$ & $90.35\%$ &
CE & $3.00\times10^{-2}$ & $3.07\%$ \\

CE & $1.00\times10^{-2}$ & $82.99\%$ &
& & \\

\bottomrule
\end{tabular}%
}

\end{table}

\paragraph{Token Divergence is Much More Complex.}
Token-level traces show that output divergence is not determined by the numerical error magnitude alone. 
When the baseline model assigns very similar scores to its top two candidates for the next token, even a small numerical error can change which token is selected. This then changes the subsequent generated tokens as well. Having said that, if the logits of top few tokens are extremely close, then replacing one with another will not cause significant change in the model's utility. In our experiments, smaller mean true errors were generally associated with fewer divergent outputs and higher semantic similarity. Therefore, the error bound provides a useful indicator of whether masking is likely to impact the quality of the  generated output.

\subsection{Prompt Reconstruction Attack}
\label{subsec:prompt-reconstruction-attack}
To demonstrate that inputs offloaded to the GPU need to be masked and to determine an appropriate value of $\epsilon$ for the mask, we perform a prompt reconstruction attack.   
Recall that the information available to the GPU
includes all (masked) inputs to the offloaded linear components, as well as the architecture and weights of the learned model. However, any processing done at the TDX guest, including data in its private memory, is not available to the GPU. 

We use a prompt reconstruction attack inspired by the generative embedding inversion attack of Li et al.~\cite{li-etal-2023-sentence}. Unlike their use of pooled sentence embedding, our attack uses token-level intermediate representations that may retain prompt information. This attack can be launched by the GPU using a public dataset and a local copy of the LLM to train a reconstruction model that maps offloaded inputs back to prompts.

\subsubsection{Training the Reconstruction Attacker}
\label{subsubsec:training-reconstruction-attacker}

To train the prompt reconstruction attack model, we use the Pile dataset~\cite{gao2020pile}, which is a diverse dataset for training large language models. Specifically, we use 2,000 prompts divided into 1,800 training prompts and 200 validation prompts. Each prompt is passed through the LLM, and the inputs exposed at a selected outsourced linear operator are recorded together with the original prompt. For the noise-aware attack model, the corresponding noise mechanism and privacy setting (value of $\epsilon$) are applied to these intermediate representations before they are used for training. These intermediate input-prompt pairs are then used to train a generative reconstruction model. 

The reconstruction attacker uses a pretrained \llama{} model to reconstruct the original prompt from a Layer~0 intermediate representation of the victim model. Representations exposed at any outsourced operator can be used as input. Since these hidden states are not normal token embeddings, a small trainable projection network maps them to the attacker model (see \suppmat \ for details). \cref{fig:prompt-reconstruction} shows the training process. The trained attacker is then evaluated on unseen Databricks Dolly prompts using GPU exposed inputs, and its generated sequence is taken as the reconstructed prompt.

\begin{figure}
\centering
\resizebox{0.8\columnwidth}{!}{%
\input{Figures/prompt-reconstruction}
}
\caption{Training of the prompt reconstruction model based on~\cite{li-etal-2023-sentence}. The masked input $\widetilde{X}$ is projected into a soft prefix $\widetilde{X}'$, and passed to the attacker LLM (\llama{}). During training, teacher forcing uses tokens from original prompt $X$.}
\label{fig:prompt-reconstruction}
\end{figure}

\subsubsection{Attacker Settings}
\label{subsubsec:prompt-reconstruction-attacker-settings}

We evaluate the reconstruction attack under several attacker settings: (1)  \emph{Noise-free attacker:} which is a reconstruction model trained and evaluated on non-masked inputs, (2) \emph{Noise-oblivious attacker:} which applies the reconstruction model from the noise-free setting to masked inputs, and (3) \emph{Noise-aware attacker:} who further trains the model by generating auxiliary inputs using the noise mechanism with the same privacy setting (value of $\epsilon$) as the protected system.  
These attacks can be further divided into how many components are taken into account by the attack model: (1)  \emph{Single Component:} attack is launched on a single component, and it is most potent when applied to the $(Q, K, V)$ component since this component is closest to the original prompt, and (2) \emph{Composite attacker:} which uses multiple outsourced components from the same transformer layer. We combine the inputs to the $(Q, K, V)$ projections, the $O$ projection, and Gate/Up projections. The same procedure is used for both noise-free and masked settings. We provide additional details on how these inputs are combined and processed in \suppmat.

\subsubsection{Reconstruction Attack Results}
\label{subsubsec:reconstruction-attack-results}

To compare the reconstructed and original prompts, we use token precision, recall, F1-score, and semantic similarity. Precision measures how many generated tokens appear in the original prompt, while recall measures how many original tokens are recovered; F1 provides a balanced measure of both. We evaluate these metrics on 500 Dolly prompts containing 100-300 tokens.

\cref{tab:layer0-reconstruction-results} presents the Layer~0 results. In the noise-free setting, $(Q,K,V)$ gives the strongest reconstruction, with token precision, recall, and F1-score around 70\% and 79.87\% semantic similarity. LM-head reconstruction is lower overall, likely because it appears later in the inference pipeline. We therefore use $(Q,K,V)$ as the single-component attacker setting in the remaining analysis.
The composite attacker which receives inputs to three components, $(Q, K, V)$, $O$ projection and Gate/Up projections does not fare significantly better than the attack with a single component. This shows that the components exposed later in the inference pipeline provide limited additional information for reconstruction. All in all, the results show that without masking the offloaded inputs, the GPU can recover significant information. A reconstructed example is shown in \cref{tab:reconstruction_example}.

\begin{table}[t]
\centering
\caption{Prompt-reconstruction results at Layer~0 under the evaluated attacker settings.}
\label{tab:layer0-reconstruction-results}
\resizebox{\columnwidth}{!}{%
\begin{tabular}{lllcccc}
\toprule
\textbf{Condition} & \textbf{Attacker setting} & \textbf{Captured operators} & \textbf{Token precision} & \textbf{Token recall} & \textbf{Token F1} & \textbf{Semantic similarity} \\
\midrule
\multirow{3}{*}{\textbf{Without noise}} 
& Noise-free & QKV only & 71.77\% & 72.12\% & 70.30\% & 79.87\% \\
& Noise-free & LM Head & 45.40\% & 75.75\% & 55.15\% & 77.13\% \\
& Composite & QKV + O + Gate/Up & 75.25\% & 69.14\% & 71.20\% & 82.19\% \\
\midrule
\multirow{6}{*}{\textbf{With noise}} 
& Noise-oblivious & QKV only & 6.5\% & 11.06\% & 7.99\% & 5.33\% \\
& Noise-aware (EP1) & QKV only & 13.14\% & 21.81\% & 15.91\% & 5.29\% \\
& Noise-aware (EP1) & LM head & 12.63\% & 20.78\% & 15.23\% & 4.68\% \\
& Composite noise-aware (EP1) & QKV + O + Gate/Up & 12.72\% & 20.69\% & 15.30\% & 5.23\% \\
& Composite noise-aware (EP1000) & QKV + O + Gate/Up & 10.86\% & 17.98\% & 13.16\% & 6.55\% \\
& Composite noise-aware (EP2500) & QKV + O + Gate/Up & 71.61\% & 55.92\% & 62.79\% & 78.41\% \\
\bottomrule
\end{tabular}
}
\end{table}

In comparison, both the single component and composite attackers perform significantly worse on noisy inputs than in the noise-free setting, with semantic similarity remaining around $5\%$. The noise aware attacker performs only slightly better than the noise-oblivious attacker, which shows that the knowledge of the noise mechanism provides  marginal benefit. 

Although the noise-aware composite attacker achieves 12.72\% token precision and 20.69\% recall, this does not necessarily indicate meaningful reconstruction because unrelated prompts can share tokens. To examine this, we compared all $\binom{500}{2}=124{,}750$ pairs among the same 500 evaluation prompts. The average token precision, recall, and F1 were 22.02\%, 22.02\%, and 20.71\%, showing that this level of overlap does not indicate meaningful reconstruction. We also evaluate $\epsilon=1000$ and $\epsilon=2500$ to examine when reconstruction becomes meaningful. Semantic similarity remains low at 6.55\% for $\epsilon=1000$, but rises to 78.41\% at $\epsilon=2500$. These experiments are not intended to recommend such large $\epsilon$ values, but to illustrate that the much lower value of $\epsilon=1$ is indeed safe for prompt privacy.

\subsubsection{Prompt Reconstruction Example and Restoring Prompt via Post-Processing}
\label{app:reconstruction-example}

\cref{tab:reconstruction_example} shows an original prompt, along with a noise-free and noise-aware reconstructions. 
The raw reconstructions produced by the attacker often contain repeated fragments, corrupted words, and broken sentence structure. To examine whether these outputs still contain recoverable information that is not fully reflected by the raw reconstruction metrics, we perform a separate text-restoration step on all reconstructed prompts in the evaluation set. The restoration model receives only the text produced by the attacker and is instructed to correct corrupted wording, merge repeated fragments, and improve readability without adding new information or answering the prompt. We use GPT-5.4~\cite{openai_gpt54}, a highly capable general-purpose language model, for this step. Then we compare the restored text with the original prompt using the same metrics used for the initial reconstruction. The system prompt used for restoration is available in \suppmat, and is designed to keep the model focused on recovering information already present in the reconstructed text.

For the noise-free QKV reconstruction, the restoration step substantially improves the recovered text, increasing token precision from 71.77\% to 89.20\%, token F1 from 70.30\% to 79.03\%, and semantic similarity from 79.87\% to 89.54\%. This shows that the raw noise-free reconstruction contains additional recoverable information that is partly obscured by corrupted wording. In contrast, restoration provides no meaningful improvement for the noise-aware QKV attacker at $\epsilon=1$. Token precision increases only slightly from 13.14\% to 14.93\% and token F1 slightly decreases from 15.91\% to 15.76\%. Semantic similarity also shows no improvement after restoration. This indicates that the noised reconstruction does not contain substantial prompt information that can be recovered through post-processing, even when a highly capable language model is used for restoration. 

\begin{table}
\centering
\caption{Example of the original prompt, its QKV-only Layer~0 reconstructions under noise-free and noise-aware conditions, and the corresponding GPT-5.4 restored outputs.}
\label{tab:reconstruction_example}
\begin{tabularx}{\textwidth}{@{}X@{}}
\toprule
\textbf{Prompt}
Extract the ingredients in bee pollen from the text.
\textit{Context:} Bee pollen, also known as bee bread and ambrosia, is a ball or pellet of field-gathered flower pollen packed by worker honeybees, and used as the primary food source for the hive. It consists of simple sugars, protein, minerals and vitamins, fatty acids, and a small percentage of other components. Bee pollen is stored in brood cells, mixed with saliva, and sealed with a drop of honey. Bee pollen is harvested as food for humans and marketed as having various, but yet unproven, health benefits.
\\
\midrule
\end{tabularx}

\begin{tabularx}{\textwidth}{@{}XX@{}}
\textbf{Reconstructed Prompt (noise-free)} &
\textbf{Restored Prompt (noise-free) - GPT-5.4} 
\\
Extract ingredients in bee pollen from the text. Context: B pollen also as bread ambia is ball pellet fieldather flower packed workerbe, used the food for hive It of sugars protein, a or of ged pollen by honeyes and with, acids and, a percentage other. pollen stored bro cells mixed saliva, a as primary source the. consists simple, with, minerals vitamins fatty, acids as various but unven health. benefits. pollenstored bro cells combined saliva, as small of components Bee is as for humans marketed having, sealed a of components Beer is inod, with saliva and with drop honey Bee, benefits
&
Extract ingredients in bee pollen from the text. Context: Bee pollen, also known as bee bread and ambrosia, is ball or pellet field-gathered flower pollen packed by worker bees, used as the food source for the hive. It consists of simple sugars, protein, a large percentage of gathered pollen by honeybees, mixed with nectar, fatty acids, and minerals, and a small percentage of other components. Bee pollen is stored in brood cells, combined with saliva, and sealed with a drop of honey. Bee pollen is used as food for humans and marketed as having various but unproven health benefits.
\\
\bottomrule
\end{tabularx}

\begin{tabularx}{\textwidth}{@{}XX@{}}
\textbf{Reconstructed Prompt (noise-aware, $\epsilon=1$)} &
\textbf{Restored Prompt (noise-aware, $\epsilon=1$) - GPT-5.4}
\\
The of the is a of the, and it's not just because they both. Both are in their own way to be loved for who they, but also have been by others for reasons that may or not. of the has always been a bit different from other. They're more independent than some, but still need love care. of on the hand very much people around them, needs support and love. both of these qualities make them special, but can also lead problems if not handled right. of is often seen as the perfect partner, someone who will stand you through thick and thin, never leave your side, and always there when needed. However, this can sometimes lead to dependency, where one becomes too reliant another, losing independence and autonomy. of, on the other hand, is often portrayed as the bad boy, someone with attitude and who doesn't play games. However can also be hurtful, leaving person feeling rejected abandoned. of all, however, is the most complicated. They so much like each other, yet at same time don't want to lose individuality. This can lead to conflicts, where try to balance their relationship with their own needs desires. of is especially difficult, as they so willing help each other out, but at same
&
The is a of the, and it's not just because they both. Both are, in their own way, to be loved for who they are, but also have been by others for reasons that may or may not. This has always been a bit different from the others. They're more independent than some, but still need love and care. On the other hand, very much needs the people around them, and needs support and love. Both of these qualities make them special, but can also lead to problems if not handled right. This is often seen as the perfect partner, someone who will stand by you through thick and thin, never leave your side, and always be there when needed. However, this can sometimes lead to dependency, where one person becomes too reliant on another, losing their independence and autonomy. , on the other hand, is often portrayed as the bad boy, someone with attitude and confidence who doesn't play games. However, this can also be hurtful, leaving the other person feeling rejected and abandoned. Most of all, however, is the most complicated. They are so much like each other, yet at the same time don't want to lose their individuality. This can lead to conflicts, where they try to balance their relationship with their own needs and
\\
\bottomrule
\end{tabularx}
\end{table}

\section{Discussions, Limitations and Future Directions}
Our experiments leverage the Databricks Dolly dataset~\cite{conover2023dolly}, which has been chosen as it was readily available to us and contained enough prompts of approximately the same size. We chose to use a public dataset not containing personal information for ethical reasons. In particular, our scheme still considers the entire prompt as private, and protects each textual entity without distinction of type or sensibility. Hence, our results can be extrapolated to any type of textual prompt.

Our scheme uses the fact that the current generation of TEEs does not provide hardware-isolated GPU execution. This might change in the future, as we mentioned in \cref{subsec:tee}. However, we expect that the computational performance of such TEEs may still fall substantially short of that of unprotected state-of-the-art GPUs, particularly for large-scale LLM inference. 

Precomputing the mask and noise correction reduces inference time at the expense of storage space. 
    Depending on the number of prompts, and the size of the model, this can exceed the memory specification. This drawback is similar to preprocessing in secure multiparty computation, where increased storage is required to trade off online computational time. Likewise, this issue also applies to Slalom~\cite{tramer2018slalom}. One way to alleviate this issue is to keep unused noise in an encrypted form on the untrusted disk, and then only load and decrypt it inside the TDX when required, similar to what is done in~\cite{tramer2018slalom}. Through a separate process, the TEE could continuously generate noise, encrypt and store it in the untrusted disk for future use.

Our global sensitivity analysis is worst-case. In particular, it does not take into account the range of values that may be taken by embeddings from actual prompts. For instance, we use the fact that normalized vectors have a Euclidean norm of less than $\sqrt{d}$ where $d$ is the embedding dimension. In practice, the norm may be well below this value. Thus, it may be possible to further reduce the scale of required noise, e.g., by examining local or smooth sensitivity~\cite{nissim2007smooth} of these inputs. This is suggested by our reconstruction attack results wihch show that even very high values of $\epsilon$ are effective in protecting the prompt.

Our prompt reconstruction attack may not be the most powerful attack possible, in terms of the capabilities of the underlying large language model used in the attack. Apart from this, the attack's assumptions are the strongest possible: full knowledge of the learned parameters and architecture of the victim model, as well as the differentially private mechanism and its parameters. However, since we add fresh noise to each input to the GPU, it is unlikely that a more powerful LLM will perform substantially better than ours.

\section{Conclusion}
\label{sec:Conclusion}

In this work, we show that splitting LLM inference between a trusted but slower TEE and an untrusted but faster GPU, with intermediate representations protected by differential privacy, enables the system to benefit from GPU acceleration while simultaneously protecting the prompt. We have shown that these representations can leak information about the original prompt, which motivates protecting the values sent outside the TEE. Our method does come with a trade-off. We show that uncontrolled differentially private noise can result in irreversible floating-point errors. However, unlike systems that use encryption to mask inputs, this error can be reined in by tuning the differential privacy parameter to balance privacy and accuracy of the LLM response. With the suggested values of the privacy parameter, we show that prompt reconstruction fails even with full knowledge of the LLM and the privacy mechanism. Beyond the use case considered in this paper, our method could also enable systems to leverage newly available, state-of-the-art GPUs without exposing information about portentially sensitive prompts to them.

\bibliographystyle{ACM-Reference-Format}
\bibliography{references}

\appendix

\end{document}

%% file: Figures/system-architecture.tex
\begin{tikzpicture}

\node (P1) {\includegraphics[width=16mm]{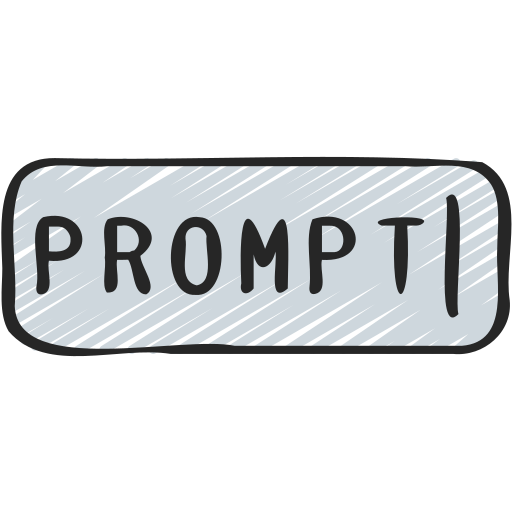}};

\node[right=3.5cm of P1] (P2) {};

\node[right=5cm of P2] (P3) {};

\draw[
    dashed,
    draw=teeborder,
    fill=teebg
]
    ([xshift=-1cm,yshift=1cm]P2.center)
    rectangle
    ++(2cm,-3cm);

\node at (P2) {\includegraphics[width=12mm]{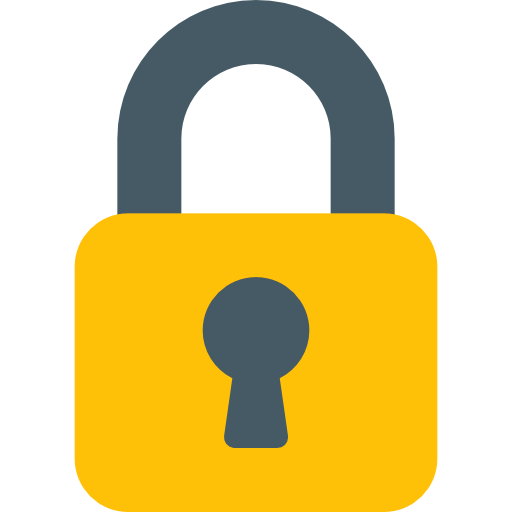}};
\node[below=15pt of P2] {TEE};

\draw[
    draw=red,
    dashed,
    fill=red!10,
]
    ([xshift=-1cm,yshift=1cm]P3.center)
    rectangle
    ++(2cm,-3cm);

\node at (P3) {\includegraphics[width=12mm]{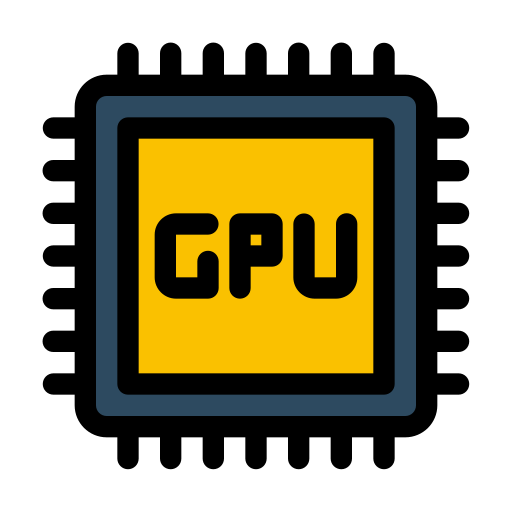}};
\node[below=15pt of P3] {GPU};

\draw[rounded corners]
    ([xshift=-1.2cm,yshift=1.2cm]P2.center)
    rectangle
    ([xshift=1.2cm,yshift=-2.7cm]P3.center);

\node[below right=2cm and 2cm of P2] (P4) {Server};

\node[above right=0.4cm and 1.2cm of P2] (P5) {to linear layers};

\draw[->]
    ([yshift=3pt]P1.east) --
    node[midway, above] {prompt}
    ([xshift=-1.2cm, yshift=3pt]P2.west);

\draw[<-]
    ([yshift=-3pt]P1.east) --
    node[midway, below] {output}
    ([xshift=-1.2cm, yshift=-3pt]P2.west);

\draw[->]
    ([xshift=1cm, yshift=10pt]P2.east) --
    ([xshift=-1cm, yshift=10pt]P3.west);

\draw[<-]
    ([xshift=1cm, yshift=0pt]P2.east) --
    ([xshift=-1cm, yshift=0pt]P3.west);

\draw[->]
    ([xshift=1cm, yshift=-16pt]P2.east) --
    node[midway, above] {$\vdots$}
    ([xshift=-1cm, yshift=-16pt]P3.west);

\draw[<-]
    ([xshift=1cm, yshift=-26pt]P2.east) -- node[midway, below] {DP protection}
    ([xshift=-1cm, yshift=-26pt]P3.west);

\end{tikzpicture}

%% file: Figures/transformer-architecture.tex
\begin{tikzpicture}

\node (X) {$X$};

\node[
    draw=gray,
    fill=gray!10,
    rounded corners,
    right=0.5cm of X,
    minimum height=0.8cm
] (embed) {Embedding};

\draw[-{Latex[length=2mm,width=2mm]}] (X) -- (embed);

\node (titleEMabove) [above=2.25cm of embed.center] {};
\node (titleEMbelow) [below=1.97cm of embed.center] {};

\node[
    draw=gray,
    dashed,
    thick,
    rounded corners,
    inner sep=8pt,
    fit=(embed) (titleEMabove) (titleEMbelow),
    label={[anchor=north west, align=left]north west:Embedding\\Model},
] (EM) {};

\node (forkMO) [right=0.5cm of embed] {};

\node[
    draw=brown,
    fill=brown!10,
    rounded corners,
    right=1cm of embed,
    minimum height=0.8cm
] (norm) {Norm};

\draw[-{Latex[length=2mm,width=2mm]}] (embed) -- (norm);

\node[
        draw=brown,
        fill=brown!10,
        rounded corners,
        right=1cm of norm,
        yshift=1.5cm,
        minimum height=0.8cm
    ] (Q) {$Q$};

    \node[
        draw=brown,
        fill=brown!10,
        rounded corners,
        right=1cm of norm,
        minimum height=0.8cm
    ] (K) {$K$};

    \node[
        draw=brown,
        fill=brown!10,
        rounded corners,
        right=1cm of norm,
        yshift=-1.5cm,
        minimum height=0.8cm
    ] (V) {$V$};

\node at (Q.north east) {\includegraphics[width=8mm]{Figures/gpu-2.png}};

\node at (K.north east) {\includegraphics[width=8mm]{Figures/gpu-2.png}};

\node at (V.north east) {\includegraphics[width=8mm]{Figures/gpu-2.png}};

    \node (fork) [right=0.25cm of norm] {};

    \draw (norm.east) -- (fork.center);

    \draw[-{Latex[length=2mm,width=2mm]}] (fork.center) |- (Q.west);
    \draw[-{Latex[length=2mm,width=2mm]}] (fork.center) -- (K.west);
    \draw[-{Latex[length=2mm,width=2mm]}] (fork.center) |- (V.west);

    \node[
    draw=brown,
    fill=brown!10,
    rounded corners,
    right=0.5cm of K,
    minimum height=0.8cm
    ] (S) {Softmax};

\draw[-{Latex[length=2mm,width=2mm]}] (Q.east) -| (S.north);
\draw[-{Latex[length=2mm,width=2mm]}] (K.east) -- (S.west);
\draw[-{Latex[length=2mm,width=2mm]}] (V.east) -| (S.south);

\node[
    draw=brown,
    fill=brown!10,
    rounded corners,
    right=0.5cm of S,
    minimum height=0.8cm
    ] (O) {$O$};

\node at (O.north east) {\includegraphics[width=8mm]{Figures/gpu-2.png}};

\draw[-{Latex[length=2mm,width=2mm]}] (S.east) -- (O.west);

\node[
    draw=brown,
    fill=brown!10,
    rounded corners,
    right=0.5cm of O,
    minimum height=0.8cm
    ] (residual) {Residual};

\draw[-{Latex[length=2mm,width=2mm]}] (O.east) -- (residual.west);

\node (bend0) [below=2.25cm of forkMO.center] {};

\node (bend1) [below=2cm of fork.center] {};
\node (bend2) [below=2.25cm of residual.center] {};

\draw[-{Latex[length=2mm,width=2mm]}] (forkMO.center) -- (bend0.center) -- (bend2.center) -- (residual.south);

\node[
    draw=brown,
    dashed,
    thick,
    rounded corners,
    inner sep=8pt,
    fit=(Q) (K) (V) (norm) (residual),
    label={[anchor=north east]north east:Self Attention Mechanism}
] (AM) {};

\node[
    draw=mygreen,
    fill=mygreen!10,
    rounded corners,
    right=0.7cm of residual,
    minimum height=0.8cm
    ] (norm2) {Norm};

\draw[-{Latex[length=2mm,width=2mm]}] (residual.east) -- (norm2.west);

\node[
    draw=mygreen,
    fill=mygreen!10,
    rounded corners,
    right=0.5cm of norm2,
    minimum height=0.8cm
    ] (gateup) {Gate/Up};

\draw[-{Latex[length=2mm,width=2mm]}] (norm2.east) -- (gateup.west);

\node at (gateup.north east) {\includegraphics[width=8mm]{Figures/gpu-2.png}};

\node[
    draw=mygreen,
    fill=mygreen!10,
    rounded corners,
    right=0.5cm of gateup,
    minimum height=0.8cm
    ] (silu) {SiLU};

\draw[-{Latex[length=2mm,width=2mm]}] (gateup.east) -- (silu.west);

\node[
    draw=mygreen,
    fill=mygreen!10,
    rounded corners,
    right=0.5cm of silu,
    minimum height=0.8cm
    ] (down) {Down};

\draw[-{Latex[length=2mm,width=2mm]}] (silu.east) -- (down.west);

\node[
    draw=mygreen,
    fill=mygreen!10,
    rounded corners,
    right=0.5cm of down,
    minimum height=0.8cm
    ] (residual2) {Residual};

\draw[-{Latex[length=2mm,width=2mm]}] (down.east) -- (residual2.west);

 \node (forkFF) [right=0.23cm of residual] {};
 \node (anchorkFF1) [below=1cm of forkFF.center] {};
  \node (anchorFF2) [below=1cm of residual2.center] {};

\draw[-{Latex[length=2mm,width=2mm]}] (forkFF.center) -- (anchorkFF1.center) -- (anchorFF2.center) -- (residual2.south);

\node (bend3) [below=1.67cm of norm2.center] {};
\node (bend4) [above=1.67cm of norm2.center] {};

\node[
    draw=mygreen,
    dashed,
    thick,
    rounded corners,
    inner sep=8pt,
    fit=(norm2) (bend3) (bend4) (residual2),
    label={[anchor=north east]north east:Feed-Forward Network}
] (FF) {};

\node (title) [above=2.25cm of norm.center] {};

\node[
    draw=myblue,
    dashed,
    thick,
    rounded corners,
    inner sep=8pt,
    fit= (norm) (title) (AM) (FF),
    label={[anchor=north west]north west:Transformer Block (one of $L$ layers)},
] (TF) {};

\node[
    draw=red,
    fill=red!10,
    rounded corners,
    right=1cm of residual2,
    minimum height=0.8cm
    ] (norm3) {Norm};

\draw[-{Latex[length=2mm,width=2mm]}] (residual2.east) -- (norm3.west);

\node[
    draw=red,
    fill=red!10,
    rounded corners,
    right=0.5cm of norm3,
    minimum height=0.8cm,
    align=center
    ] (lm) {Language\\Model};

\draw[-{Latex[length=2mm,width=2mm]}] (norm3.east) -- (lm.west);

\node at (lm.north east) {\includegraphics[width=8mm]{Figures/gpu-2.png}};

\node[
    right=0.5cm of lm,
    align=center
    ] (nexttoken) {Next\\Token};

\draw[-{Latex[length=2mm,width=2mm]}] (lm.east) -- (nexttoken.west);

\node (titleLMabove) [above=2.25cm of norm3.center] {};
\node (titleLMbelow) [below=1.97cm of norm3.center] {};

\node[
    draw=red,
    dashed,
    thick,
    rounded corners,
    inner sep=8pt,
    fit= (norm3) (titleLMabove) (titleLMbelow) (lm),
    label={[anchor=north west]north west:Language Modelling},
] {};

\end{tikzpicture}

%% file: Figures/methodology.tex
\begin{tikzpicture}

\pgfdeclarelayer{background}
\pgfsetlayers{background,main}

\node[
    draw=gray,
    fill=gray!10,
    ellipse,
    minimum height=0.8cm,
    minimum width=1.4cm
] (x) {$\vtr{x}$};

\node[
    draw=gray,
    fill=gray!10,
    ellipse,
    below right = 0.35cm and 0.20cm of x,
    minimum height=0.8cm,
    minimum width=1.4cm
] (bank) {$\left\{\bm{\eta},W\bm{\eta}\right\}$};

\node[
    right = 0 pt of bank, 
] (noisebank){Noise Bank};

\node[
    draw=gray,
    fill=gray!10,
    ellipse,
    below = 1cm of x,
    minimum height=0.8cm,
    minimum width=1.4cm
] (wx) {$W\vtr{x}$};

\node[
    draw=gray,
    fill=gray!10,
    rounded corners,
    right = 1.5cm of x,
    minimum height=0.8cm,
    minimum width= 3cm
] (xplusn) {\shortstack{\text{Noise addition} \\ $\widetilde{\vtr{x}} = \vtr{x} + \bm{\eta}$}};

\node[
    draw=gray,
    fill=gray!10,
    rounded corners,
    right = 1.5cm of wx,
    minimum height=0.8cm,
    minimum width= 3cm
] (wxminuswn) {\shortstack{\text{Noise Correction} \\ $W\vtr{x} = W\widetilde{\vtr{x}} - W\bm{\eta}$}};

\draw[-{Latex[length=2mm,width=2mm]}] (x) -- (xplusn);
\draw[-{Latex[length=2mm,width=2mm]}] (bank) -- (xplusn);
\draw[{Latex[length=2mm,width=2mm]}-] (wx) -- (wxminuswn);
\draw[-{Latex[length=2mm,width=2mm]}] (bank) -- (wxminuswn);

\begin{pgfonlayer}{background}
\node[
    draw=teeborder,
    fill=teebg,
    dashed,
    thick,
    rounded corners,
    inner sep=8pt,
    fit=(x) (wx) (wxminuswn)(noisebank), %
] (NB) {};
\end{pgfonlayer}

\node[
    below = 5 pt of NB, 
] (tee){Trusted Execution Environment};

\node[
    draw=gray,
    fill=gray!10,
    ellipse,
    right = 1.5cm of xplusn,
    minimum height=0.8cm,
    minimum width=1.4cm
] (xtilde) {$\widetilde{\vtr{x}}$};

\node[
    draw=gray,
    fill=gray!10,
    ellipse,
    right = 1.5cm of xplusn,
    right = 1.5cm of wxminuswn,
    minimum height=0.8cm,
    minimum width=1.4cm
] (wxtilde) {$W\widetilde{\vtr{x}}$};

\draw[-{Latex[length=2mm,width=2mm]}] (xplusn) -- (xtilde);
\draw[{Latex[length=2mm,width=2mm]}-] (wxminuswn) -- (wxtilde);

\begin{pgfonlayer}{background}
\node[
    draw=brown,
    fill=brown!10,
    dashed,
    thick,
    rounded corners,
    inner sep=8pt,
    fit=(xtilde) (wxtilde),%
] (SM) {};
\end{pgfonlayer}

\node[
    below = 5 pt of SM, 
] (smzone){Shared Memory};

\node[
    draw=gray,
    fill=gray!10,
    rounded corners,
    below right = 0.25cm and 1.5cm of xtilde,
    minimum height=0.8cm,
    minimum width= 3cm
] (wxcomp) {$W\widetilde{\vtr{x}} = W(\vtr{x} + \bm{\eta})$};

\node[
    above = 18pt of wxcomp,
    inner sep=0pt,
] (noisyinput){Noisy Input};

\node[
    below = 17pt of wxcomp, 
    inner sep=0pt,
] (noisyoutput){Noisy Output};

\node (bend) [right=1.5cm of wxtilde] {};

\draw[-{Latex[length=2mm,width=2mm]}] (xtilde) -| (wxcomp);

\draw[-{Latex[length=2mm,width=2mm]}] (wxcomp) |- (wxtilde);

\begin{pgfonlayer}{background}
\node[
    draw=red,
    fill=red!10,
    dashed,
    thick,
    rounded corners,
    inner sep=8pt,
    fit=(wxcomp) (bend) (noisyinput) (noisyoutput) %
] (GPU) {};

\node[
    below = 5 pt of GPU, 
] (){Untrusted GPU};
\end{pgfonlayer}

\end{tikzpicture}

%% file: Figures/prompt-reconstruction.tex
\begin{tikzpicture}

\node (X) {$\widetilde{X}$};
\node[
    below=0.1cm of  X,
    align=center
] (masked)  {Masked\\Input};

\node[
    draw=gray,
    fill=gray!10,
    rounded corners,
    right=1cm of X,
    minimum height=0.8cm
] (expand) {Expand};

\draw[-{Latex[length=2mm,width=2mm]}] (X) -- (expand);

\node[
    draw=gray,
    fill=gray!10,
    rounded corners,
    right=0.5cm of expand,
    minimum height=0.8cm
] (silu) {SiLU};

\draw[-{Latex[length=2mm,width=2mm]}] (expand) -- (silu);

\node[
    draw=gray,
    fill=gray!10,
    rounded corners,
    right=0.5cm of silu,
    minimum height=0.8cm
] (reduce) {Reduce};

\draw[-{Latex[length=2mm,width=2mm]}] (silu) -- (reduce);

\node[
    draw=gray,
    fill=gray!10,
    rounded corners,
    right=0.5cm of reduce,
    minimum height=0.8cm
] (silu2) {SiLU};

\draw[-{Latex[length=2mm,width=2mm]}] (reduce) -- (silu2);

\node[
    draw=gray,
    fill=gray!10,
    rounded corners,
    right=0.5cm of silu2,
    minimum height=0.8cm
] (norm) {Norm};

\draw[-{Latex[length=2mm,width=2mm]}] (silu2) -- (norm);

\node (titlePNabove) [above=1cm of expand.center] {};
\node (titlePNbelow) [below=1cm of expand.center] {};

\node[
    draw=gray,
    dashed,
    thick,
    rounded corners,
    inner sep=8pt,
    fit=(expand) (titlePNabove) (titlePNbelow) (norm),
    label={[anchor=north west, align=left]north west:Projection Network},
] (PN) {};

\node[
    right=1cm of norm,
    minimum height=0.8cm
] (softX) {$\widetilde{X}'$};

\node[
    below=0.1cm of  softX,
    align=center
] (softpf)  {Soft\\Prefix};

\draw[-{Latex[length=2mm,width=2mm]}] (norm) -- (softX);

\node[
    right=1cm of softX,
    minimum height=0.8cm
] (llama) {};

\draw[-{Latex[length=2mm,width=2mm]}] (softX) -- (llama);

\node (titleLabove) [above=1cm of llama.center] {};
\node (titleLbelow) [below=1cm of llama.center] {};
\node (titleLright) [right=1cm of llama.center] {};

\node[
    draw=brown,
    dashed,
    thick,
    rounded corners,
    inner sep=8pt,
    fit=(llama) (titleLabove) (titleLbelow) (titleLright),
    label={[anchor=north west, align=left]north west:Attacker\\LLM},
] (L) {};

\node[
    right=1cm of titleLright,
    align=center
    ] (nexttoken) {Next\\Token};

\draw[-{Latex[length=2mm,width=2mm]}] (titleLright) -- (nexttoken);

\node (forkC) [below=1.7cm of nexttoken] {};

\node[
    left = 0.5cm of forkC,
] (cross) {\includegraphics[width=5mm]{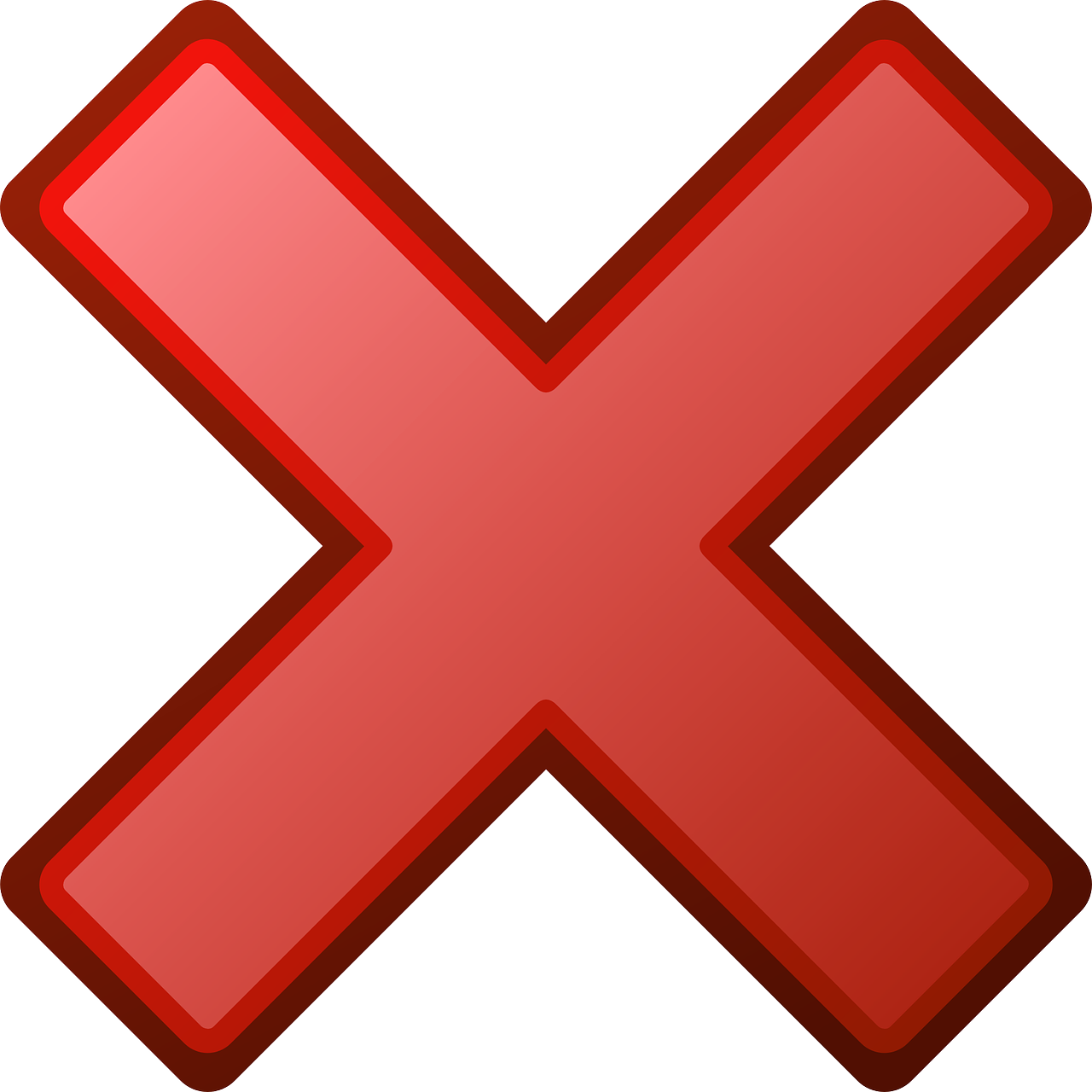}};

\draw[-] (nexttoken) -- (forkC.center) -- (cross);

\node[
   below=0.35cm of L,
   minimum height=0.8cm
] (teacherX) {$X$};

\node[
   left=0.18cm of teacherX,
   minimum height=0.8cm
] (teacherF) {Original Tokens (Teacher Forcing)};

\draw[-{Latex[length=2mm,width=2mm]}] (teacherX.north) -- (L);

\end{tikzpicture}